%% file: main.tex
\documentclass[runningheads]{llncs}
\usepackage[colorlinks=true]{hyperref}
\usepackage{enumitem}
\usepackage{amsmath}
\usepackage{amssymb}
\usepackage{mathtools}
\usepackage[official]{eurosym}
\usepackage{xcolor}
\usepackage{pgfplots}
\pgfplotsset{compat=1.13} 
\usepackage{float}
\usepackage{color}
\usepackage{tikz}
\usetikzlibrary{arrows.meta}
\usetikzlibrary{arrows,automata}
\usetikzlibrary{shapes,decorations,patterns,positioning}
\usepackage{makecell}
\usepackage{subcaption}
\usepackage{tikz-3dplot}
\usepackage{listings}
\usepackage{multirow}
\usepackage{changepage}
\usepackage{marginnote}
\def\esat{{\sat}}

\title{Multi-player Equilibria Verification for Concurrent Stochastic Games}
\author{Marta Kwiatkowska\inst{1} \and Gethin Norman\inst{2} \and David~Parker\inst{3} \and Gabriel Santos\inst{1}}
\institute{Department of Computing Science, University of Oxford, UK
\and
School of Computing Science, University of Glasgow, UK
\and School of Computer Science, University of Birmingham, UK}

\input{macros-dave.tex}

\begin{document}

\maketitle
\begin{abstract}
Concurrent stochastic games (CSGs) are an ideal formalism for modelling probabilistic systems that feature
multiple players or components with distinct objectives
making concurrent, rational decisions.
Examples include communication or security protocols and multi-robot navigation.
Verification methods for CSGs exist but are limited to scenarios where
agents or players are grouped into two \emph{coalitions},
with those in the same coalition sharing an identical objective.
In this paper, we propose \emph{multi-coalitional} verification techniques for CSGs.
We use subgame-perfect social welfare (or social cost) optimal Nash equilibria,
which are strategies where there is no incentive for any coalition
to unilaterally change its strategy in any game state,
and where the total combined objectives are maximised (or minimised).
We present an extension of the temporal logic rPATL
(probabilistic alternating-time temporal logic with rewards)
to specify equilibria-based properties for any number of distinct coalitions,
and a corresponding model checking algorithm
for a variant of stopping games.
We implement our techniques in the PRISM-games tool and apply them to several case studies, including a secret sharing protocol and a public good game.

\end{abstract}

\input{introduction}

\input{background}

\input{logic}
\input{model_checking}

\input{experiments}
\input{conclusions}

\clearpage
\bibliographystyle{splncs03.bst}
\bibliography{bib}

\appendix
\input{correctness}
\clearpage

\end{document}

%% file: macros-dave.tex
\renewcommand{\emptyset}{\varnothing}

\DeclareMathOperator{\opt}{opt}

\usepackage{scalerel}

\newcommand\scale[2]{\vstretch{#1}{\hstretch{#1}{#2}}}

\newcommand{\appref}[1]{Appendix~\ref{#1}}
\newcommand{\sectref}[1]{Section~\ref{#1}}

\newcommand{\figref}[1]{Figure~\ref{#1}}

\newcommand{\egref}[1]{Example~\ref{#1}}

\newcommand{\eqnref}[1]{(\ref{#1})}

\newcommand{\propref}[1]{Proposition~\ref{#1}}
\newcommand{\lemref}[1]{Lemma~\ref{#1}}
\newcommand{\defref}[1]{Definition~\ref{#1}}

\newcommand{\assumref}[1]{Assumption~\ref{#1}}

\newcommand{\defdefref}[2]{Definitions~\ref{#1} and \ref{#2}}

\spnewtheorem{assumption}{Assumption}{\bfseries}{\itshape}
\newcounter{exampcount}
\newenvironment{examp}
{\refstepcounter{exampcount}
\vskip6pt\noindent
{\bf Example \arabic{exampcount}.}}
{}
\newcommand{\startpara}[1]{{%
\vskip6pt\noindent
{\bf #1.}}}

\newcommand{\startparai}[1]{{%
\vskip6pt\noindent
{\em #1.}}}

\def\ra{{\rightarrow}}

\def\cC{{\mathcal{C}}}

\def\Nset{\mathbb{N}}
\def\Qset{\mathbb{Q}}
\def\Rset{\mathbb{R}}

\def\Eset{\mathbb{E}}

\def\ra{\rightarrow} %
\def\rmdef{\,\stackrel{\mbox{\rm {\tiny def}}}{=}}
\newcommand{\true}{\mathtt{true}} %

\renewcommand{\leq}{\leqslant}
\renewcommand{\geq}{\geqslant}

\def\squareforqed{\hbox{\rlap{$\sqcap$}$\sqcup$}}
\def\qed{\ifmmode\squareforqed\else{\unskip\nobreak\hfil
\penalty50\hskip1em\null\nobreak\hfil\squareforqed
\parfillskip=0pt\finalhyphendemerits=0\endgraf}\fi}

\newcommand\game{{\sf G}}
\newcommand\nfgame{{\sf N}}

\newcommand\sinit{{\bar{s}}}

\newcommand\dist{{\mathit{Dist}}}

\newcommand\Prob{{\mathit{Prob}}}

\newcommand\val{{\mathit{val}}}

\newcommand{\last}{\mathit{last}}
\newcommand{\ipaths}{\mathit{IPaths}}
\newcommand{\fpaths}{\mathit{FPaths}}

\def\AP{{\mathit{AP}}}
\def\sat{{\,\models\,}}

\def\Sat{{\mathit{Sat}}}
\newcommand{\coalition}[1]{\langle \! \langle {#1} \rangle \! \rangle}
\def\notsat{{\,\not\models\,}}

\def\next{{X\,}}
\def\until{{\ {\cal U}\ }}
\def\buntil{{\ {\cal U}^{\leq k}\ }}

\def\next{{\mathtt X\,}}
\def\until{{\ \mathtt{U}\ }}
\def\untilop{{\mathtt{U}}}

\def\buntil{{\ \mathtt{U}^{\leq k}\ }}

\def\future{{\mathtt{F}\ }}
\def\futureop{{\mathtt{F}}}

\newcommand{\scumul}[1]{\mathtt{C}^{#1}} %

\newcommand{\ap}{\mathsf{a}}
\newcommand{\sinstant}[1]{\mathtt{I}^{#1}} %

\newcommand\V{{\mathtt V}}
\newcommand\probopP{{\mathtt P}}
\newcommand\nashop[3]{\coalition{#1}_{#2}(#3)}

\newcommand\probop[2]{\probopP_{#1}[\,{#2}\,]}

\newcommand\rewopR{{\mathtt R}}
\newcommand\rewop[3]{\rewopR^{#1}_{#2}[\,{#3}\,]}

\newcommand{\lab}{{\mathit{L}}}

\newcommand{\rew}{\mathit{rew}}

%% file: introduction.tex
\section{Introduction}\label{intro-sect}

Stochastic multi-player games are a modelling formalism
that involves a number of players making sequences of rational decisions,
each of which results in a probabilistic change in state. 
They are well suited to modelling systems that feature
competitive or collaborative behaviour between multiple components or agents,
operating in uncertain or stochastic environments.
Examples include communication or security protocols,
which may employ randomisation or send messages over unreliable channels,
and multi-robot or multi-vehicle navigation,
where sensors and actuators are subject to noise or prone to failure.
A game-theoretic approach to modelling also allows rewards,
incentives or resource usage to be incorporated.
For example, mechanism design can be used to create protocols
reliant on incentive schemes to improve robustness against
selfish behaviour by participants, as
utilised in network routing protocols~\cite{RT02} and auctions~\cite{CSS07}.

Designing reliable systems that comprise multiple components
with differing objectives is a challenge.
This is further complicated by the need to consider stochastic behaviour.
Formal verification techniques for stochastic multi-player games %
can be a valuable tool for tackling this problem.
The probabilistic model checker PRISM-games~\cite{KNPS20}
has been developed for modelling and analysis of stochastic games:
both the turn-based variant, where one player makes a decision in
each state, and the concurrent variant, where players make decisions
concurrently and without knowledge of each other's actions.
PRISM-games also supports strategy synthesis, which allows automated generation of
strategies for one or more players in the game, 
which are guaranteed to satisfy quantitative correctness specifications
written in temporal logic.

The temporal logics used in PRISM-games for stochastic games
are based on rPATL (probabilistic alternating-time temporal logic with rewards)~\cite{CFK+13b},
which combines features of the game logic ATL~\cite{AHK02}
and probabilistic temporal logics.
For example, in a 3-player game, the formula
$\coalition{\mathtt{rbt_1}}\probop{\geq p}{\future \mathsf{g}_1}$ states
``robot 1 has a strategy under which the probability of it successfully
reaching its goal is at least $p$, regardless of the strategies of robots 2 and 3''.
Model checking and strategy synthesis algorithms for rPATL exist
for both turn-based~\cite{CFK+13b} and concurrent stochastic games~\cite{KNPS18}. 

rPATL uses ATL's coalition operator $\coalition{\cdot}$ to formulate properties.
In the above example, there are two coalitions, one containing robot 1
and the other robots 2 and 3. %
The coalitions have distinctly opposing (zero-sum) objectives,
aiming either to maximise or minimise the probability of robot 1 reaching its goal.
A recent extension~\cite{KNPS19} allows the two coalitions to
have distinct objectives, using Nash equilibria.
More precisely, it uses subgame-perfect social welfare
optimal Nash equilibria, which are strategies for all players
where there is no incentive for either coalition to unilaterally change
its strategy in any state, %
and where the total combined objectives are maximised.
For example,
$\nashop{\mathtt{rbt_1}{:}\mathtt{rbt_2},\mathtt{rbt_3}}{\max=?}{\probop{}{\future \mathsf{g}_1}\ {+}\ \probop{}{\future (\mathsf{g}_2\wedge\mathsf{g}_3)}}$
asks for such an equilibrium, where the two coalitions' objectives are to maximise the probability of reaching their own (distinct) goals. Model checking rPATL for both the zero-sum~\cite{CFK+13b,KNPS18} and
equilibria-based~\cite{KNPS19} properties
has the advantage that it essentially reduces to the analysis of
2-player stochastic games, for which various algorithms exist
(e.g.,~\cite{AHK07,AM04,CAH13}).
However, a clear limitation %
is the assumption that agents can, or would be
willing to, collaborate and form two distinct coalitions.

In this paper, we propose \emph{multi-coalitional} verification techniques for concurrent stochastic games (CSGs).
We extend the temporal logic rPATL to allow reasoning about any number of distinct coalitions with different quantitative objectives,
expressed using a variety of temporal operators
capturing either the probability of an event occurring or a reward measure. 
We then give a model checking algorithm for the logic against CSGs,
restricting our attention to a variant of stopping games~\cite{CFKSW13},
which, with probability 1, eventually reach a point where
the outcome of each player's objective does not change by continuing.
Our algorithm uses a combination of backward induction (for finite-horizon operators)
and value iteration (for infinite-horizon operators). 
A key ingredient of the computation is finding optimal Nash equilibria for $n$-player games, which we perform using support enumeration \cite{PNS04} and a mixture of SMT and non-linear optimisation solvers.
We implement our techniques in the PRISM-games tool and apply them to several case studies, including a secret sharing protocol and a public good game.
This allows us to verify multi-player scenarios that could not
be analysed with existing techniques~\cite{KNPS19}.

\startpara{Related work}
As summarised above, there are various algorithms to solve CSGs,
e.g.,~\cite{AHK07,AM04,CAH13}, and model checking techniques have been developed for
both zero-sum~\cite{KNPS18} and equilibria-based~\cite{KNPS19} versions of rPATL
on CSGs, implemented in PRISM-games~\cite{KNPS20}.
However, all of this work assumes or reduces to the 2-player case.
Equilibria for $n$-player CSGs are considered in~\cite{BBG+19},
but only complexity results, not algorithms, are presented.
Other tools exist to reason about equilibria, including 
PRALINE~\cite{BRE13}, EAGLE~\cite{TGW15}, EVE~\cite{GNP+18},
MCMAS-SLK~\cite{CLM+14} (via strategy logic) and Gambit~\cite{GAMB},
but these are all for \emph{non-stochastic} games.

%% file: background.tex
\section{Preliminaries}\label{prelim-sect}

We let $\dist(X)$ denote the set of probability distributions over set $X$. For any vector $v$ we use $v(i)$ to denote the $i$th entry of the vector.
\begin{definition}[Normal form game] 
A (finite, $n$-person) \emph{normal form game} (NFG) is a tuple $\nfgame = (N,A,u)$ where: $N=\{1,\dots,n\}$ is a finite set of players; $A = A_1 {\times} \cdots {\times} A_n$ and $A_i$ is a finite set of actions available to player $i \in N$; $u {=} (u_1,\dots,u_n)$ and $u_i : A \rightarrow \Rset$ is a utility function for player $i \in N$.
\end{definition}
For an NFG $\nfgame$, the players choose actions at the same time, where the choice for player $i \in N$ is over the action set $A_i$. When each player $i$ chooses $a_i$, the utility received by player $j$ equals $u_j(a_1,\dots,a_n)$. A (mixed) strategy $\sigma_i$ for player $i$ is a distribution over its action set. Let $\eta_{a_i}$ denote the pure strategy that selects action $a_i$ with probability 1 and $\Sigma^i_\nfgame$ the set of strategies for player $i$. 
A \emph{strategy profile} $\sigma{=} (\sigma_1,\dots,\sigma_n)$ is a tuple of strategies for each player and under $\sigma$ the expected utility of player $i$ equals:
\[ 
\begin{array}{c}
u_i(\sigma) \ \rmdef \ \sum_{(a_1,\dots,a_n) \in A} u_i(a_1,\dots,a_n) \cdot \left( \prod_{j=1}^n \sigma_j(a_j) \right) \, .
\end{array} 
\]
For strategy $\sigma_i$ of a player, the \emph{support} of $\sigma_i$ is the set of actions it chooses with nonzero probability, i.e., $\{ a_i \in A_i \mid \sigma_i(a_i){>}0 \}$. Furthermore, the support of a profile is the product of the supports of the individual strategies and a profile is said to have full support if it includes all available action tuples.

We now fix an NFG $\nfgame{=} (N,A,u)$ and introduce the notion of Nash equilibrium and the variants we require. For profile $\sigma{=}(\sigma_1,\dots,\sigma_n)$ and player $i$ strategy $\sigma_i'$, we define the sequence $\sigma_{-i} \rmdef (\sigma_1,\dots,\sigma_{i-1},\sigma_{i+1},\dots,\sigma_n)$ and profile $\sigma_{-i}[\sigma_i'] \rmdef (\sigma_1,\dots,\sigma_{i-1},\sigma_i',\sigma_{i+1},\dots,\sigma_n)$.
\begin{definition}[Best and least response] 
For player $i$ and strategy sequence $\sigma_{-i}$, 
a \emph{best response} for player $i$ to $\sigma_{-i}$ is a strategy $\sigma^\star_i$ for player $i$ such that $u_i(\sigma_{-i}[\sigma^\star_i]) \geq u_i(\sigma_{-i}[\sigma_i])$ for all $\sigma_i \in \Sigma^i_\nfgame$ and
a \emph{least response} for player $i$ to $\sigma_{-i}$ is a strategy $\sigma^\star_i$ for player $i$ such that $u_i(\sigma_{-i}[\sigma^\star_i]) \leq u_i(\sigma_{-i}[\sigma_i])$ for all $\sigma_i \in \Sigma^i_\nfgame$.
\end{definition}
\begin{definition}[Nash equilibrium]\label{nash-def}
A strategy profile $\sigma^\star$ is a \emph{Nash equilibrium} (NE) if $\sigma_i^\star$ is a best response to $\sigma_{-i}^\star$ for all $i \in N$.
\end{definition}
\begin{definition}[Social welfare NE]\label{swne-def}
An NE $\sigma^\star$ is a \emph{social welfare optimal NE} (SWNE) and $\langle u_i(\sigma^\star) \rangle_{i \in N}$ are \emph{SWNE values} if  $u_1(\sigma^\star){+}\cdots$ ${+}u_n(\sigma^\star)\geq u_1(\sigma){+} \cdots+u_n(\sigma)$ for all NE $\sigma$ of $\nfgame$.
\end{definition}
\begin{definition}[Social cost NE]\label{scne-def}
A profile $\sigma^\star$ of $\nfgame$ is a \emph{social cost optimal NE} (SCNE) and $\langle u_i(\sigma^\star) \rangle_{i \in N}$ are \emph{SCNE values} if $\sigma^\star$ is an NE of $\nfgame^{-}= (N,A,{-}u)$ and $u_1(\sigma^\star){+}\cdots{+}u_n(\sigma^\star)\leq u_1(\sigma){+} \cdots{+}u_n(\sigma)$ for all NE $\sigma$ of $\nfgame^{-}$. Furthermore, $\sigma^\star$ is an SWNE of $\nfgame^{-}$ if and only if $\sigma^\star$ is an SCNE of $\nfgame$.
\end{definition}
The notion of SWNE is standard~\cite{NRTV07} and applies when utility values represent profits or rewards.
We use the dual notion of SCNE for utilities that represent losses or costs. Example objectives in this category include
minimising the probability of a fault  %
or the expected time to complete a task.
We have chosen to represent SCNE directly since this is more natural than the alternative of simply negating utilities, particularly in the case of probabilities.

\begin{examp}\label{nfg-eg}
Consider the NFG representing a variant of a \emph{public good} game~\cite{HHCN19}, in which three players each receive a fixed amount of capital ($10$\euro) and can choose to invest none, half or all of it in a common stock (represented by the actions $\mathit{in}^0_i$, $\mathit{in}^5_i$ and $\mathit{in}^{10}_i$ respectively). The total invested by the players is multiplied by a factor $f$ and distributed equally among the players, and the aim of the players is to maximise their profit. The utility function of player $i$ is therefore given by:
\[
u_i(\mathit{in}^{k_1}_1,\mathit{in}^{k_2}_2,\mathit{in}^{k_3}_3) =  (f/3) {\cdot} (k_1 {+} k_2 {+} k_3) - k_i \, .
\]
for $k_i \in \{0,5,10\}$ and $1{\leq} i {\leq} 3$.
If $f{=}2$, then the profile where each investor chooses not to invest is an NE and each player's utility equals $0$. More precisely, if a single player was to deviate from this profile by investing half or all of their capital, then their utility would decrease to $(2/3) {\cdot} 5 {-} 5 = - 5/3$ or $(2/3) {\cdot} 10 {-} 10 = -10/3$, respectively. Since this is the only NE it is also the only SWNE and $(0,0,0)$ are the only SWNE values. The profile where each player invests all of their capital is not an NE as, under this profile, a player's utility equals $(2/3) {\cdot} 30 {-} 10 = 10$ and any player can increase their utility to $(2/3) {\cdot} 25 {-} 5 = 35/3$ by deviating and investing half of their capital.

On the other hand, if $f{=}3$, then there are two NE: when all players invest either none or invest all of their capital. The sum of utilities of the players under these profiles are $0{+}0{+}0=0$ and $20{+}20{+}20=60$ respectively, and therefore the second profile is the only SWNE.
\end{examp}

\begin{definition}[Concurrent stochastic game] A \emph{concurrent stochastic multi-player game} (CSG) is a tuple
$\game = (N, S, \bar{S}, A, \Delta, \delta, \AP, \lab)$ where:
\begin{itemize}
\item $N=\{1,\dots,n\}$ is a finite set of players;
\item $S$ is a finite set of states and $\bar{S} \subseteq S$ is a set of initial states;
\item $A = (A_1\cup\{\bot\}) {\times} \cdots {\times} (A_n\cup\{\bot\})$ where $A_i$ is a finite set of actions available to player $i \in N$ and $\bot$ is an idle action disjoint from the set $\cup_{i=1}^n A_i$;
\item $\Delta \colon S \rightarrow 2^{\cup_{i=1}^n A_i}$ is an action assignment function;
\item $\delta \colon S {\times} A \rightarrow \dist(S)$ is a (partial) probabilistic transition function;
\item $\AP$ is a set of atomic propositions and $\lab \colon S \rightarrow 2^{\AP}$ is a labelling function.
\end{itemize}
\end{definition}
A CSG $\game$ starts in an initial state $\sinit \in \bar{S}$ and, when in state $s$, each player $i \in N$ selects an action from its available actions $A_i(s) \rmdef \Delta(s) \cap A_i$ if this set is non-empty, and from $\{ \bot \}$ otherwise. For any state $s$ and action tuple $a=(a_1,\dots,a_n)$, the partial probabilistic  transition function $\delta$ is defined for $(s,a)$ if and only if $a_i \in A_i(s)$ for all $i \in N$. We augment CSGs with \emph{reward structures}, which are tuples of the form $r{=}(r_A,r_S)$ where $r_A : S {\times} A \ra \Rset$ and $r_S : S \ra \Rset$ are action and state reward functions, respectively.

A \emph{path} is a sequence $\pi = s_0 \xrightarrow{\alpha_0} s_1 \xrightarrow{\alpha_1} \cdots$ such that $s_i \in S$, $\alpha_i = (a^i_1,\dots,a^i_n) \in A$, $a^i_j \in A_j(s_i)$ for $j \in N$ and $\delta(s_i,\alpha_i)(s_{i+1}){>}0$ for all $i {\geq} 0$. Given a path $\pi$, we denote by $\pi(i)$ the $(i{+}1)$th state, $\pi[i]$ the $(i{+}1)$th action, and if $\pi$ is finite, $\last(\pi)$ the final state. %
The sets of finite and infinite paths (starting in state $s$) of $\game$ are given by $\fpaths_\game$ and $\ipaths_\game$ ($\fpaths_{\game,s}$ and $\ipaths_{\game,s}$).

\emph{Strategies} are used to resolve the choices of the players. Formally, a strategy for player $i$ is a function $\sigma_i \colon \fpaths_{\game} \ra \dist(A_i \cup \{ \bot \})$ such that, if $\sigma_i(\pi)(a_i){>}0$, then $a_i \in A_i(\last(\pi))$. A \emph{strategy profile} is a tuple $\sigma {=} (\sigma_1,\dots,\sigma_{n})$ of strategies for all players. The set of strategies for player $i$ and set of profiles are denoted $\Sigma^i_\game$ and $\Sigma_\game$. Given a profile $\sigma$ and state $s$, let $\ipaths^\sigma_{\game,s}$ denote the infinite paths with initial state $s$ corresponding to $\sigma$.
We can then define, using standard techniques~\cite{KSK76}, a probability measure $\Prob^{\sigma}_{\game,s}$ over $\ipaths^{\sigma}_{\game,s}$ and, for a random variable $X \colon \ipaths_{\game} \rightarrow \Rset$, the expected value $\Eset^{\sigma}_{\game,s}(X)$ of $X$ in $s$ under $\sigma$.

In a CSG, a player's utility or \emph{objective} is represented by a random variable $X_i \colon \ipaths_{\game} \rightarrow \Rset$. Such variables can encode, for example, the probability of reaching a target or the expected cumulative reward before reaching a target. Given an objective for each player, social welfare and social cost NE can be defined as for NFGs. As in \cite{KNPS19}, we consider \emph{subgame-perfect} NE~\cite{OR04}, which are NE in \emph{every state} of the CSG.
In addition, for infinite-horizon objectives, the existence of NE is an open problem~\cite{BMS14} so, for such objectives, we use $\varepsilon$-NE,
which exist for any $\varepsilon{>}0$. Formally, we have the following definition.
\begin{definition}[Subgame-perfect $\varepsilon$-NE] For CSG $\game$ and $\varepsilon{>}0$, a profile $\sigma^\star$ is a \emph{subgame-perfect $\varepsilon$-NE}
for the objectives $\langle X_i \rangle_{i \in N}$ if and only if:\/ $\Eset^{\sigma^\star}_{\game,s}(X_i) \geq \sup_{\sigma_i \in \Sigma_i} \Eset^{\sigma^\star_{-i}[\sigma_i]}_{\game,s}(X_i) - \varepsilon$ for all $i \in N$ and $s \in S$.
\end{definition}

\begin{examp}\label{csg-eg}
We now extend \egref{nfg-eg} to allow the players to invest their capital (and subsequent profits) over a number of months and assume that, at the end of each month, the parameter $f$ can either increase or decrease by $0.2$ with probability $0.1$. %
This can be modelled as a CSG $\game$ whose states are tuples of the form $(m,f,c_1,c_2,c_3)$, where $m$ is the current month, $f$ the parameter value and $c_i$ is the current capital of player $i$ (the initial capital plus or minus any profits or losses made in previous months). If $f$ has initial value $2$ and the players start with a capital of 10\euro, then the initial state of $\game$ equals $(0,2,10,10,10)$. The actions of player $i$ are of the form $\mathit{in}^{k_i}_i$, which corresponds to $i$ investing $k_i$ in the current month. The probabilistic transition function of the game is such that:
\begin{eqnarray*}
\lefteqn{\hspace*{-2cm}\delta((m,f,c_1,c_2,c_3),(\mathit{in}^{k_1}_1,\mathit{in}^{k_2}_2,\mathit{in}^{k_3}_3))(m',f',c_1',c_2',c_3') } \\ &=&
\left\{ \begin{array}{cl}
0.8 & \ \mbox{if $m'{=}m{+}1$, $f'{=}f$ and $c_i'{=} c_i {+} p_i$} \\
0.1 & \ \mbox{if $m'{=}m{+}1$, $f'{=}f{+}0.2$ and $c_i'{=} c_i {+} p_i$} \\
0.1 & \ \mbox{if $m'{=}m{+}1$, $f'{=}f{-}0.2$ and $c_i'{=} c_i {+} p_i$} \\
0 & \ \mbox{otherwise}
\end{array} \right.
\end{eqnarray*}
where $p_i =  (f/3) {\cdot} (k_1 {+} k_2 {+} k_3) - k_i$ for $k_i \in \{0,5,10\}$ and $1{\leq} i {\leq} 3$.

If we are interested in the profits of the players after $k$ months, then we can consider a random variable for player $i$ which would return, for a path with $(k{+}1)$th state $(k,f,c_1,c_2,c_3)$, the value $c_i{-}10$.
\end{examp}

%% file: logic.tex
\section{Extended rPATL with Nash Formulae}\label{logic-sec}

We now consider the logic rPATL with Nash formulae~\cite{KNPS19} and enhance it with equilibria-based properties that can separate players into more than two coalitions.
\begin{definition}[Extended rPATL syntax]
The syntax of our extended version of {\rm rPATL} is given by the grammar:
\begin{eqnarray*}
\phi & \; \coloneqq \; & \mathtt{true} \mid \ap \mid \neg \phi \mid \phi \wedge \phi \mid \coalition{C}\probop{\sim q}{\psi} \mid \coalition{C}\rewop{r}{\sim x}{\rho} \mid \nashop{C_1{:}\cdots{:}C_m}{\opt \sim x}{\theta} \\
\theta & \; \coloneqq \; & \probop{}{\psi}{+}{\cdots}{+}\probop{}{\psi} \ \mid \  \rewop{r}{}{\rho}{+}{\cdots}{+}\rewop{r}{}{\rho}  \\
\psi & \; \coloneqq \; & \next \phi \ \mid \ \phi \buntil \phi \ \mid \ \phi \until \phi \\
\rho & \; \coloneqq \; &  \sinstant{=k} \ \mid \ \scumul{\leq k} \ \mid \ \future \phi
\end{eqnarray*}
where $\ap$ is an atomic proposition, $C$ and $C_1,\dots,C_m$ are coalitions of players such that $C_i \cap C_j = \emptyset$ for all $1\leq i \neq j \leq m$ and $\cup_{i=1}^m C_i = N$, $\opt \in \{ \min,\max\}$, $\sim \,\in \{<, \leq, \geq, >\}$, $q \in \Qset\cap[0, 1]$, $x \in \Qset$, $r$ is a reward structure and $k \in \Nset$. %
\end{definition} %
Our addition to the logic is
\emph{Nash formulae} of the form $\nashop{C_1{:}{\cdots}{:}C_m}{\opt \sim x}{\theta}$, where the \emph{nonzero sum} formulae $\theta$ comprises a sum of $m$ probability or reward objectives
(for full details of the rest of the logic see~\cite{KNPS18,KNPS19}).
The formula $\nashop{C_1{:}{\cdots}{:}C_m}{\max \sim x}{\probop{}{\psi_1}{+}{\cdots}{+}\probop{}{\psi_m}}$ holds in a state
if, when the players form the coalitions $C_1,\dots,C_m$,
there is a subgame-perfect SWNE for which the \emph{sum} of the values of the objectives
$\probop{}{\psi_1},\dots,\probop{}{\psi_m}$ for the coalitions $C_1,\dots,C_m$
satisfies ${\sim} x$.
The case for reward objectives is similar and, for formulae of the form $\nashop{C_1{:}{\cdots}{:}C_m}{\min \sim x}{\theta}$,
we require the existence of an SCNE rather than an SWNE.
We also allow \emph{numerical} queries of the form $\nashop{C_1{:}{\cdots}{:}C_m}{\opt =?}{\theta}$, which return the sum of the SWNE or SCNE values.

In a probabilistic nonzero-sum formula $\theta{=}\probop{}{\psi_1}{+}{\cdots}{+}\probop{}{\psi_m}$,
each objective $\psi_i$ can be a next ($\next \phi$), bounded until ($\phi_1 \buntil  \phi_2$) or until ($\phi_1 \until \phi_2$) formula, with the usual equivalences, e.g., $\future\phi \equiv \true \until \phi$. For the reward case $\theta{=}\rewop{r_1}{}{\rho_1}{+}{\cdots}{+}\rewop{r_m}{}{\rho_m}$,
each $\rho_i$ refers to a reward formula with respect to reward structure $r_i$ and can be bounded instantaneous reward ($\sinstant{=k}$), bounded accumulated reward ($\scumul{\leq k}$) or reachability reward ($\future \phi$).
\begin{examp}
Recall the public good CSG from \egref{csg-eg}. Examples of nonzero-sum formulae in our logic include:
\begin{itemize}
\item
$\nashop{\mathit{p}_1{:}\mathit{p}_2{:}\mathit{p}_3}{\max \geq 3}{\probop{}{ \future c_1{\geq}20}{+}\probop{}{\future c_2{\geq}20}{+}\probop{}{\future c_3{\geq}20}}$ states that the three players can collaborate such that they each eventually double their capital with probability 1;
\item
$\nashop{\mathit{p}_1{:}\mathit{p}_2{:}\mathit{p}_3}{\max=?}{\rewop{\mathit{cap}_1}{}{\sinstant{=4}}{+}\rewop{\mathit{cap}_2}{}{\sinstant{=4}}{+}\rewop{\mathit{cap}_3}{}{\sinstant{=4}}}$ asks for the sum of the expected capital of the players at $4$ months when they collaborate, where the state reward function of $\mathit{cap}_i$ returns the capital of player $i$.
\item
$\nashop{\mathit{p}_1{:}\mathit{p}_2{:}\mathit{p}_3}{\max\geq 50}{\rewop{\mathit{pro}_1}{}{\scumul{\leq 6}}{+}\rewop{\mathit{pro}_2}{}{\scumul{\leq 6}}{+}\rewop{\mathit{pro}_3}{}{\scumul{\leq 6}}}$ states that the sum of the expected cumulative profit of the players after 6 months when they collaborate is at least 50, where the action reward function of $\mathit{pro}_i$ returns the expected profit of player $i$ from a state for the given action tuple.
\end{itemize}
\end{examp}
\noindent
In order to give the semantics of the logic,
we require an extension of the notion of \emph{coalition games}~\cite{KNPS18} which, given a CSG $\game$ and partition $\cC$ of the players into $m$ coalitions, reduces $\game$ to an $m$-player coalition game, where each player corresponds to one of the coalitions in $\cC$. Without loss of generality, we assume $\cC$ is of the form $\{ \{1,\dots,n_1\}, \{ n_1{+}1,\dots n_2 \}, \dots, \{ n_{m-1}{+}1, \dots n_{m} \}\}$ and let $j_\cC$ denote player $j$'s position in its coalition. 
\begin{definition}[Coalition game] For CSG $\game{=}(N, S, \bar{s}, A, \Delta, \delta, \AP, \lab)$ and partition of the players into $m$ coalitions $\cC=\{C_1, \dots, C_{m}\}$, we define the \emph{coalition game} $\game^\cC {=} ( M, S, \bar{s}, A^\cC, \Delta^\cC, \delta^\cC, \AP, \lab)$ as an $m$-player CSG where:
\begin{itemize}
    \item $M = \{1,\dots,m \}$;
	\item $A^\cC = (A^\cC_1\cup \{ \bot\}) {\times} \cdots {\times} (A^\cC_{m}\cup \{ \bot\})$;
	\item $A^\cC_i = (\prod_{j \in C_i} (A_j\cup\{\bot\}) \setminus \{(\bot,\dots,\bot)\} \big)$ for all $i \in M$;
	\item for any $s \in S$ and $i \in M\!:$ $a_i^\cC \in \Delta^\cC(s)$ if and only if either $\Delta(s) \cap A_j =\emptyset$ and $a_i^\cC(j_\cC)=\bot$ or $a_i^\cC(j_\cC) \in \Delta(s)$  for all $j \in C_i$;
	\item for any $s \in S$ and $(a^\cC_1,\dots,a^\cC_{m})\in A^\cC\!:$ $\delta^\cC(s,(a^\cC_1,\dots,a^\cC_{m})) = \delta(s,(a_1,\dots,a_n))$ where for $i \in M$ and $j \in C_i$ if $a_i^\cC {=} \bot$, then $a_j {=}\bot$ and otherwise $a_j {=} a_i^\cC(j_\cC)$.
\end{itemize}
Furthermore, for a reward structure $r=(r_A,r_S)$, by abuse of notation we use $r=(r^{\cC}_A,r^{\cC}_S)$ for the corresponding reward structure of $\game^{\cC}$ where:
\begin{itemize}
    \item for any $s \in S$, $a^C_i \in A^C_i\!:$ $r^{\cC}_{A^{\cC}}(s, (a_1^C,\dots,a_{m}^C))=r_A(s,(a_1,\dots,a_n))$ where for $i \in M$ and $j \in C_i$, if $a_i^\cC = \bot$, then $a_j{=}\bot$ and otherwise $a_j {=} a_i^\cC(j_\cC)$;
    \item for any $s \in S\!:$ $r^{\cC}_S(s) {=} r_S(s)$.
\end{itemize}
\end{definition}
The logic includes infinite-horizon objectives ($\untilop$, $\futureop$),
for which the existence of SWNE and SCNE is open~\cite{BMS14}.
However, $\varepsilon$-SWNE and $\varepsilon$-SCNE \emph{do} exist for any $\varepsilon>0$.

\begin{definition}[Extended rPATL semantics]\label{sem-def}
For a CSG $\game$, $\varepsilon>0$ and a formula $\phi$,
the satisfaction relation $\esat$ is defined inductively over the structure of~$\phi$.
The propositional logic fragment $(\mathtt{true}$, $\ap$, $\neg$, $\wedge)$
is defined in the usual way. The zero-sum formulae $\coalition{C} \probop{\sim q}{\psi}$ and $\coalition{C} \rewop{r}{\sim x}{\rho}$ are defined as in~\emph{\cite{KNPS18,KNPS19}}.
For a Nash formula and state $s \in S$ in CSG $\game$, we have:
\begin{eqnarray*}
s \esat \nashop{C_1{:}\cdots{:}C_m}{\opt \sim x}{\theta} & \Leftrightarrow &
\exists \sigma^\star \in \Sigma_{\game^{\cC}} . \, \left( \Eset^{\sigma^\star}_{\game^{\cC},s}(X^\theta_1)+ \cdots + \Eset^{\sigma^\star}_{\game^{\cC},s}(X^\theta_m) \right) \sim x
\end{eqnarray*}
and $\sigma^\star=(\sigma_1^\star,\dots,\sigma_m^\star)$ 
is a subgame perfect $\varepsilon$-SWNE if $\opt{=}\max$, and a subgame perfect $\varepsilon$-SCNE$\,$ if $\opt{=}\min$, for the objectives $(X^\theta_1,\dots,X^\theta_m)$ in $\game^{\cC}$ where $\cC = \{ C_1,\dots,C_m\}$ and
for $1{\leq}i{\leq}m$ and $\pi \in \ipaths_{\game^\cC,s}^{\sigma^\star}:$
\begin{eqnarray*}
X^{\probop{}{\psi^1}{+}\cdots{+}\probop{}{\psi^m}}_i(\pi) & \ = \ & 1 \;\mbox{if $\pi \esat \psi^i$ and 0 otherwise} \\
X^{\rewop{r_{\scale{.75}{1}}}{}{\rho^1}{+}\cdots{+}\rewop{r_{\scale{.75}{m}}}{}{\rho^m}}_i(\pi) & \ = \ &  \rew(r_i,\rho^i)(\pi) \\
\pi \esat \next \phi 
& \ \Leftrightarrow \ & 
\pi(1) \esat \phi \\
\pi \esat \phi_1 \buntil \phi_2 
& \ \Leftrightarrow \ &
\exists i \leq k . \, (\pi(i) \esat \phi_2 \wedge \forall j < i . \, \pi(j) \esat \phi_1 )
\\
\pi \esat \phi_1 \until \phi_2 
& \ \Leftrightarrow \ & 
\exists i \in \Nset . \, ( \pi(i) \esat \phi_2 \wedge \forall j < i  . \, \pi(j) \esat \phi_1 ) \\
\rew(r,\sinstant{=k})(\pi) & = & r_S(\pi(k)) \\
\rew(r,\scumul{\leq k})(\pi) & = & \mbox{$\sum_{i=0}^{k-1}$} \big( r_A(\pi(i),\pi[i])+r_S(\pi(i)) \big) \\
\rew(r,\future  \phi)(\pi) & = & \left\{ \begin{array}{cl}
\infty
& \mbox{if} \; \forall j \in \Nset . \, \pi(j) \notsat \phi \\
\mbox{$\sum_{i=0}^{k_\phi}$} \big( r_A(\pi(i),\pi[i])+r_S(\pi(i)) \big) & \mbox{otherwise}
\end{array} \right.
\end{eqnarray*}
and $k_\phi = \min \{ k{-}1 \mid \pi(k) \esat \phi \}$.
\end{definition}

%% file: model_checking.tex
\section{Model Checking CSGs against Nash Formulae}\label{mc-sect}

rPATL is a branching-time logic and so the model checking algorithm works by recursively computing the set $\Sat(\phi)$ of states satisfying formula $\phi$
over the structure of $\phi$. Therefore, to extend the existing algorithm of~\cite{KNPS18,KNPS19}, we need only consider formulae of the form $\nashop{C_1{:}{\cdots}{:}C_m}{\opt \sim x}{\theta}$. %
From \defref{sem-def}, this requires the computation of subgame-perfect SWNE or SCNE values of the objectives $(X^\theta_1,\dots,X^\theta_m)$ and a comparison of their sum to the threshold $x$. 

We first explain how we compute SWNE values in NFGs. Next we consider CSGs, and show how to compute subgame-perfect SWNE and SCNE values for finite-horizon objectives and approximate values for infinite-horizon objectives. For the remainder of this section we fix an NFG $\nfgame$ and CSG $\game$.

As in \cite{KNPS19}, to check nonzero-sum properties on CSGs,
we have to work with a restricted class of games.
This can be seen as a variant of \emph{stopping games}~\cite{CFKSW13},
as used for multi-objective turn-based stochastic games.
Compared to~\cite{CFKSW13}, we use a weaker, objective-dependent assumption,
which ensures that, under all profiles,
with probability 1, eventually the outcome of each player's objective does not change by continuing.
This can be checked using graph algorithms~\cite{dA97a}.
\begin{assumption}\label{game-assum}
For each subformula $\probop{}{\phi_1^i \until \phi_2^i}$, set $\Sat(\neg \phi_1^i \vee \phi_2^i)$ is reached with probability 1 from all states under all profiles. For each subformula $\rewop{r}{}{\future \phi^i}$, the set $\Sat(\phi^i)$ is reached with probability 1 from all states under all profiles.
\end{assumption} 

\startpara{Computing SWNE Values of NFGs}
Computing NE values for an $n$-player game is a complex task when $n{>}2$, as it can no longer be reduced to a linear programming problem. 
The algorithm for the two-player case presented in \cite{KNPS19}, based on \emph{labelled polytopes}, starts by considering all the regions of the strategy profile space and then iteratively reduces the search space as positive probability assignments are found and added as restrictions on this space. The efficiency of this approach deteriorates when analysing games with large numbers of actions and when one or more players are indifferent, as the possible assignments resulting from action permutations need to be exhausted.

Going in the opposite direction, support enumeration~\cite{PNS04} is a method for computing NE that exhaustively examines all sub-regions, i.e., supports, of the strategy profile space, one at a time, checking whether that sub-region contains equilibria.
The number of supports is exponential in the number of actions and equals $\prod_{i=1}^n (2^{|A_i|} -1)$. Therefore computing SWNE values through support enumeration will only be efficient for games with a small number of actions.  

We now show how, for a given support, using the following lemma, the computation of SWNE profiles can be encoded as a \emph{nonlinear programming problem}. The lemma states that a profile is an NE if and only if any player switching to a single action in the support of the profile yields the same utility for the player and switching to an action outside the support can only decrease its utility.
\begin{lemma}[\cite{PNS04}]\label{necond-lem} The strategy profile $\sigma{=}(\sigma_1,\dots,\sigma_n)$ of $\nfgame$ is an NE if and only if the following conditions are satisfied: 
\begin{eqnarray}
\forall i \in N . \, \forall a_i \in A_i . \, && \sigma_i(a_i)>0 \rightarrow u_i(\sigma_{-i}[\eta_{a_i}]) = u_i(\sigma) \label{necond1-eqn} \\
\forall i \in N . \, \forall a_i \in A_i . \, && \sigma_i(a_i)=0 \rightarrow u_i(\sigma_{-i}[\eta_{a_i}]) \leq u_i(\sigma) \, . \label{necond2-eqn}
\end{eqnarray}
\end{lemma}
Given the support $B = B_1 {\times} \cdots {\times}B_n \subseteq A$, to construct the problem, we first choose \emph{pivot} actions\footnote{For each $i \in N$ this can be any action in $B_i$.} $b^p_i \in B_i$ for $i \in N$, then the problem is to minimise:
\begin{equation}\label{opt-eqn}
\begin{array}{c}
\left( \sum_{i \in N} \max_{a \in A} u_i(a) \right) - \sum_{i \in N} \left( \sum_{b \in B} u_i(b) \cdot \left( \prod_{j \in N} p_{j,b_j} \right) \right)
\end{array}
\end{equation}
subject to:
\begin{eqnarray}
\!\!\!\!\!\!\!\!\!\!\mbox{$\sum_{c \in B_{-i}(b_i^p)} u_i(c) \cdot \left( \prod_{j\in N_{-i}} p_{j,c_j} \right)
- \sum_{c \in B_{-i}(b_i)} u_i(c) \cdot \left( \prod_{j\in N_{-i}} p_{j,c_j} \right)$} & = & 0 \label{eq-eqn} \\
\!\!\!\!\!\!\!\!\!\!\mbox{$\sum_{c \in B_{-i}(b^p_i)} u_i(c) \cdot \left( \prod_{j\in N_{-i}} p_{j,c_j} \right)
- \sum_{c \in B_{-i}(a_i)} u_i(c) \cdot \left( \prod_{j\in N_{-i}} p_{j,c_j} \right)$} & \geq & 0  \label{leq-eqn} \\
\mbox{$\sum_{b_i \in B_i} p_{i,b_i}$} =  1 \quad \mbox{and} \quad p_{i,b_i} &  >  &  0 \label{ge-eqn}
\end{eqnarray}
for all $i \in N$, $b_i \in B_i {\setminus} \{b_i^p \}$ and $a_i \in A_i {\setminus} B_i$ where $B_{-i}(c_i) = B_1 {\times} \cdots {\times} B_{i-1} {\times} \{ c_i \}$ ${\times} B_{i+1} {\times} \cdots {\times} B_n$ and $N_{-i} = N {\setminus} \{i\}$.
The variables %
in the above program represent the probabilities players choose different actions, i.e.\ $p_{i,b_i}$ is the probability $i$ selects $b_i$. The constraints 
\eqnref{ge-eqn} ensure the probabilities of each player sum to one and the support of the corresponding profile equals $B$. The constraints \eqnref{eq-eqn} and \eqnref{leq-eqn} require that the solution corresponds to an NE as these encode the constraints \eqnref{necond1-eqn} and \eqnref{necond2-eqn}, respectively, of \lemref{necond-lem} when restricting to pivot actions. This restriction is sufficient as \eqnref{necond1-eqn} requires all actions in the support to yield the same utility. The first term in \eqnref{opt-eqn} corresponds to the maximum possible sum of utilities for the players, i.e.\ it sums the maximum utility of each player, and the second sums the individual utilities of the players when they play according to the profile corresponding to the solution. By minimising the difference between these two terms, we require the solution to be social welfare optimal.

SMT solvers with nonlinear modules can be used to solve such problems, although they can be inefficient. Alternative approaches %
include \emph{barrier} or \emph{interior-point} methods~\cite{NWW09}.

\input{three_prisoners}

\begin{examp} Consider the instance of three prisoner's dilemma with utilities described in Table~\ref{tab:3pd} where $A_i {=}\{c_i,d_i\}$ for $1{\leq}i{\leq}3$. 
For the full support $B^\mathit{fs}$
the utility of player $i$ equals:
\[
u_i(B^\mathit{fs}) = p_{i,c_i} {\cdot} u_i(B_{-i}^\mathit{fs}(c_i)) + p_{i,d_i} {\cdot} u_i(B_{-i}^\mathit{fs}(d_i)) 
\]
where $u_i(B_{-i}^\mathit{fs}(c_i))$ and $u_i(B_{-i}^\mathit{fs}(d_i))$ are the utilities of player $i$ when switching to choosing action $c_i$ and $d_i$ with probability 1 and are given by: 
\[
\begin{array}{rcl}
u_i(B_{-i}^\mathit{fs}(c_i)) & = & 7{\cdot}p_{j,c_j}{\cdot}p_{k,c_k} +
3{\cdot}p_{j,c_j}{\cdot}p_{k,d_k} +
3{\cdot}p_{j,d_j}{\cdot}p_{k,c_k}\\
u_i(B_{-i}^\mathit{fs}(d_i)) & = & 9{\cdot}p_{j,c_j}{\cdot}p_{k,c_k} +
5{\cdot}p_{j,c_j}{\cdot}p_{k,d_k} +
5{\cdot}p_{j,d_j}{\cdot}p_{k,c_k} +
p_{j,d_j}{\cdot}p_{k,d_k} 
\end{array} 
\]
for $1 {\leq} i {\neq} j {\neq} k {\leq} 3$. Now, choosing $c_i$ as the pivot action for $1 {\leq} i {\leq} 3$, we obtain the nonlinear program of minimising:
\[
27 - (u_1({B^\mathit{fs}}) +  u_2({B^\mathit{fs}}) +  u_3({B^\mathit{fs}})) 
\]
subject to: $u_i(B^\mathit{fs}_{-i}(c_i)) {-} u_i(B^\mathit{fs}_{-1}(d_i)) = 0$, $p_{i,c_i} {+} p_{i,d_i} = 1$, $p_{i,c_i} {>} 0$ and $p_{i,d_i} {>} 0$ for $1 {\leq}i{\leq}3$.
When trying to solving this problem, we find that there is no NE as the constraints reduce to $p_{3,c_3}{\cdot} (p_{2,d_2}+1)=-1$, which cannot be satisfied. 

For the partial support $B^\mathit{ps} {=} \{ (d_1,d_2,d_3)\}$,  $d_i$ is the only choice of pivot action for player $i$ and, after a reduction, we obtain the program of minimising:
\[
27 - (p_{2,d_2}p_{3,d_3} + p_{1,d_1}p_{3,d_3} +  p_{1,d_1}p_{2,d_2})
\]
subject to: $p_{i,d_i}{\cdot}p_{j,d_j} \geq 0$, $p_{i,d_i} {=} 1$ and $p_{i,d_i} {>} 0$ for $1 {\leq} i {\neq} j {\leq} 3$. Solving this problem we see it is satisfied, and therefore the profile where each player $i$ chooses $d_i$ is an NE. This demonstrates that, as for the two-player prisoner's dilemma, defection dominates cooperation for all players, which leads to the only NE.
\end{examp}

\startpara{Computing Values of Nash Formulae}
We now show how to compute the SWNE values of a Nash formula  $\nashop{C_1{:}{\cdots}:C_m}{\opt\sim x}{\theta}$. The case for SCNE values can be computed similarly, where to compute SCNE values of a NFG $\nfgame$, we use \defref{scne-def}, negate the utilities of $\nfgame$, find SWNE values of the resulting NFG and return the negation of these values as SCNE values of $\nfgame$. %

If all the objectives in the nonzero sum formula $\theta$ are finite-horizon, \emph{backward induction}~\cite{SW+01,NMK+44} can be applied to compute (precise) subgame-perfect SWNE values. On the other hand, when all the objectives are infinite-horizon,
we extend the techniques of \cite{KNPS19} for two coalitions and use \emph{value iteration}~\cite{CH08} to approximate subgame-perfect SWNE values.
In cases when there is a combination of finite- and infinite-horizon objectives, we can extend the techniques of \cite{KNPS19} and make all objectives infinite-horizon by modifying the game in a standard manner.

{Value computation for each type of objective is described below.
The extension of the two-player case \cite{KNPS19} is non-trivial:
in that case, when one player reaches their goal, we can apply MDP verification techniques by making the players form a coalition to reach the remaining goal. However, in the $n$-player case, if one player reaches their goal we cannot reduce the analysis to an $(n{-}1)$-player game, as the choices of the player that has reached its goal can still influence the outcomes of the remaining players,  and making the player form a coalition with one of the other players will give the other player an advantage.
Instead, we need to keep track of the set of players that have reached their goal (denoted $D$) and can no longer reach their goal in the case of until formulae (denoted $E$), and define the values at each iteration using these sets.

We use the notation $\V_{\game^\cC}(s,\theta)$  ($\V_{\game^\cC}(s,\theta,n)$)  for the vector of computed values of the objectives $(X^\theta_1,X^\theta_2,\dots,X^\theta_m)$ in state $s$ of $\game^\cC$ (at iteration $n$). We also use $\mathbf{1}_m$ and $\mathbf{0}_m$ to denote a vector of size $m$ whose entries all equal to 1 or 0, respectively. For any set of states $S'$ and state $s$ we let $\eta_{S'}(s)$ equal $1$ if $s \in S'$ and $0$ otherwise. Furthermore, to simplify the presentation the step bounds appearing in path and reward formulae can take negative values. %
\startparai{Bounded Probabilistic Until}
If $\theta = \probop{}{\phi^1_1 \ \untilop^{\leq k_1} \ \phi^1_2}{+}\cdots{+}\probop{}{\phi^{m}_1 \ \untilop^{\leq k_m} \ \phi^{m}_2}$, we compute SWNE values of the objectives for the nonzero-sum formulae $\theta_n = \probop{}{\phi^1_1 \ \untilop^{\leq k_1-n}\ \phi^1_2}{+}\cdots{+}\probop{}{\phi^{m}_1 \ \untilop^{\leq k_m-n} \ \phi^{m}_2}$ for $0{\leq} n {\leq} k$ recursively, where $k=\max\{k_1,\dots,k_l\}$ and $\V_{\game^\cC}(s,\theta)=\V_{\game^\cC}(s,\emptyset,\emptyset,\theta_0)$. For any state $s$ and $0{\leq} n {\leq} k$, $D,E \subseteq M$ such that $D \cap E = \emptyset$:
\[
\V_{\game^\cC}(s,D,E,\theta_n) = \left\{ \begin{array}{cl}
(\eta_{D}(1),\dots,\eta_{D}(m)) & \;\; \mbox{if $D \cup E = M$} \\
\V_{\game^\cC}(s,D \cup D',E,\theta_n) & \;\; \mbox{else if $D' \neq \varnothing$} \\
\V_{\game^\cC}(s,D,E \cup E',\theta_n) & \;\; \mbox{else if $E' \neq \varnothing$} \\
\val(\nfgame) & \;\; \mbox{otherwise}
\end{array} \right. 
\]
where $D' =  \{l \in M {\setminus} (D\cup E) \mid s \in \Sat(\phi^l_2)\}$, $E' =  \{l \in M {\setminus} (D\cup E) \mid s \in \Sat(\neg \phi^l_1 \wedge \neg \phi^l_2)\}$ and $\val(\nfgame)$ equals SWNE values of the game $\nfgame = (M,A^\cC,u)$ in which for any $1{\leq}l{\leq}m$ and $a \in A^\cC$:
\[ 
u_l(a) = \left\{
\begin{array}{cl}
1 & \;\;\mbox{if $l \in D$} \\
0 & \;\;\mbox{else if $l \in E$} \\
0 & \;\;\mbox{else if $n_l{-}n \leq 0$} \\
\sum_{s' \in S} \delta^{\cC}(s,a)(s') \cdot v^{s',l}_{n-1} & \;\;\mbox{otherwise}
\end{array} \right.
\]
and $(v^{s',1}_{n-1},v^{s',2}_{n-1},\dots,v^{s',m}_{n-1}) = \V_{\game^\cC}(s',D,E,\theta_{n{-}1})$ for all $s' \in S$.
\startparai{Instantaneous Rewards}
If $\theta = \rewop{r_1}{}{\sinstant{= k_1}}+\cdots+\rewop{r_{m}}{}{\sinstant{= k_{m}}}$, we compute SWNE values of the objectives for the nonzero-sum formulae $\theta_{n}=\rewop{r_1}{}{\sinstant{=n_1-n}}+\dots+\rewop{r_{m}}{}{\sinstant{=n_l-n}}$ for $0{\leq} n {\leq}k$ recursively, where $k{=}\max\{k_1,\dots,k_l\}$ and $\V_{\game^{\cC}}(s,\theta) = \V_{\game^{\cC}}(s, \theta_0)$. For any state $s$ and $0{\leq} n {\leq} k$, $\V_{\game^\cC}(s,\theta_{n})$ equals SWNE values of the game $\nfgame = (M,A^{\cC},u)$ in which for any $1{\leq}l{\leq}m$ and $a \in A^{\cC}$:
\[ 
u_l(a) = \left\{
\begin{array}{cl}
0  & \;\;\mbox{if $n_l{-}n < 0$} \\
\sum_{s' \in S} \delta^{\cC}(s,a)(s') \cdot r^l_S(s')  & \;\;\mbox{else if $n_l{-}n = 0$} \\
\sum_{s' \in S} \delta^{\cC}(s,a)(s') \cdot v^{s',l}_{n+1} & \;\;\mbox{otherwise}
\end{array} \right.
\]
and $(v^{s',1}_{n+1},\dots,v^{s',m}_{n+1}) = \V_{\game^\cC}(s',\theta_{n+1})$ for all $s' \in S$.
\startparai{Bounded Cumulative Rewards}
If $\theta = \rewop{r_1}{}{\scumul{\leq k_1}}+\cdots+\rewop{r_{m}}{}{\scumul{\leq k_{m}}}$, we compute SWNE values of the objectives for the nonzero-sum formulae $\theta_{n}=\rewop{r_1}{}{\scumul{\leq n_1-n}}+\dots+\rewop{r_l}{}{\scumul{\leq n_{m}-n}}$ for $0{\leq} n {\leq} k$ recursively, where $k{=}\max\{k_1,\dots,k_l\}$ and $\V_{\game^{\cC}}(s,\theta)$ $= \V_{\game^{\cC}}(s, \theta_0)$. For any state $s$ and $0{\leq} n {\leq} k$,
$\V_{\game^\cC}(s,\theta_n)$ equals SWNE values of the game $\nfgame = (M,A^{\cC},u)$ in which for any $1{\leq}l{\leq}m$ and $a \in A^{\cC}$:
\[ 
u_l(a) = \left\{
\begin{array}{cl}
0 & \;\;\mbox{if $n_l{-}n \leq 0$} \\
r^l_S(s) + r^l_A(s,a) + \sum_{s' \in S} \delta^{\cC}(s,a)(s') \cdot v^{s',l}_{n+1} & \;\;\mbox{otherwise}
\end{array} \right.
\]
and $(v^{s',1}_{n+1},\dots,v^{s',m}_{n+1}) = \V_{\game^\cC}(s',\theta_{n+1})$ for all $s' \in S$.
\startparai{Probabilistic Until}
If $\theta = \probop{}{\phi^1_1 \until \phi^1_2}{+}\cdots{+}\probop{}{\phi^{m}_1 \until \phi^{m}_2}$, values can be computed through value iteration as the limit $\V_{\game^\cC}(s,\theta) = \lim_{n \ra \infty} \V_{\game^\cC}(s,\theta,n)$ where $\V_{\game^\cC}(s,\theta,n)=\V_{\game^\cC}(s,\emptyset,\emptyset,\theta,n)$ and for any $D,E \subseteq M$ such that $D \cap E = \emptyset$:
\[
\V_{\game^\cC}(s,D,E,\theta,n) = \left\{ \begin{array}{cl}
(\eta_{D}(1),\dots,\eta_{D}(m)) & \;\; \mbox{if $D \cup E = M$} \\
(\eta_{\Sat(\phi^1_2)}(s),\dots,\eta_{\Sat(\phi^{m}_2)}(s)) & \;\; \mbox{else if $n=0$} \\
\V_{\game^\cC}(s,D \cup D',E,\theta,n) & \;\; \mbox{else if $D' \neq \varnothing$} \\
\V_{\game^\cC}(s,D,E \cup E',\theta,n) & \;\; \mbox{else if $E' \neq \varnothing$} \\
\val(\nfgame) & \;\; \mbox{otherwise}
\end{array} \right. 
\]
where $D' =  \{l \in M {\setminus} (D\cup E) \mid s \in \Sat(\phi^l_2)\}$, $E' =  \{l \in M {\setminus} (D\cup E) \mid s \in \Sat(\neg \phi^l_1 \wedge \neg \phi^l_2)\}$ and $\val(\nfgame)$ equals SWNE values of the game $\nfgame = (M,A^\cC,u)$ in which for any $1{\leq}l{\leq}m$ and $a \in A^\cC$:
\[ 
u_l(a) = \left\{
\begin{array}{cl}
1 & \;\;\mbox{if $l \in D$} \\
0 & \;\;\mbox{else if $l \in E$} \\
\sum_{s' \in S} \delta^{\cC}(s,a)(s') \cdot v^{s',l}_{n-1} & \;\;\mbox{otherwise}
\end{array} \right.
\]
and $(v^{s',1}_{n-1},v^{s',2}_{n-1},\dots,v^{s',m}_{n-1}) = \V_{\game^\cC}(s',D,E,\theta,n{-}1)$ for all $s' \in S$.
\startparai{Expected Reachability} If $\theta = \rewop{r_1}{}{\future \phi^1}{+}\cdots{+}\rewop{r_{m}}{}{\future \phi^{m}}$, values can be computed through value iteration as the limit 
$\V_{\game^\cC}(s,\theta) = \lim_{n \ra \infty} \V_{\game^\cC}(s,\theta,n)$ where $\V_{\game^\cC}(s,\theta,n)=\V_{\game^\cC}(s,\emptyset,\theta,n)$ and for any $D \subseteq M$:
\[
\V_{\game^\cC}(s,D,\theta,n) = \left\{ \begin{array}{cl}
\mathbf{0}_{m} & \;\; \mbox{if $D = M$} \\
\mathbf{0}_{m} & \;\; \mbox{else if $n= 0$} \\
\V_{\game^\cC}(s,D \cup D',\theta,n) & \;\; \mbox{else if $D' \neq \varnothing$} \\
\val(\nfgame) & \;\; \mbox{otherwise}
\end{array} \right. 
\]
$D' = \{l \in M {\setminus} D \mid s \in \Sat(\phi^l)\}$ and $\val(\nfgame)$ equals SWNE values of the game $\nfgame = (M,A^\cC,u)$ in which for any $1{\leq}l{\leq}m$ and $a \in A^\cC$:
\[ 
u_l(a) = \left\{
\begin{array}{cl}
0 & \;\;\mbox{if $l \in D$} \\
r^l_S(s) + r^l_A(s,a) + \sum_{s' \in S} \delta^{\cC}(s,a)(s') \cdot v^{s',l}_{n-1} & \;\;\mbox{otherwise}
\end{array} \right.
\]
and $(v^{s',1}_{n-1},v^{s',2}_{n-1},\dots,v^{s',m}_{n-1}) = \V_{\game^\cC}(s',D,\theta,n{-}1)$ for all $s' \in S$.
\startpara{Strategy Synthesis} 
When performing property verification, it is usually beneficial to include \emph{strategy synthesis}, that is, construct a witness to the satisfaction of a property. When verifying a Nash formula $\nashop{C_1{:}\cdots{:}C_m}{\opt \sim x}{\theta}$, we can also return a subgame-perfect SWNE or SCNE for the objectives $(X^\theta_1,\dots,X^\theta_m)$. This is achieved by keeping track of an SWNE for the NFG solved in each state. The synthesised strategies require randomisation and memory. Randomisation is needed for NE of NFGs. Memory is required for finite-horizon properties and since choices change after a path formula becomes true or a target is reached. For infinite-horizon properties, only approximate $\varepsilon$-NE profiles are synthesised.

\startpara{Correctness and Complexity}
The proof of correctness of the algorithm can be found in \appref{correctness-app}. In the case of finite-horizon nonzero-sum formulae the correctness of the model checking algorithm follows from the fact that we use backward induction~\cite{SW+01,NMK+44}. For infinite-horizon nonzero-sum formulae the proof is based on showing that the values of the players computed during value iteration correspond to subgame-perfect SWNE or SCNE values of finite game trees, and the values of these game trees converge uniformly to the actual values of $\game^\cC$.}
The complexity of the algorithm is linear in the formula size, and finding subgame-perfect NE for reachability objectives in $n$-player games is PSPACE~\cite{BBG+19}. Value iteration requires finding all NE for a NFG in each state of the model, and computing NE of an NFG with three (or more) players is PPAD-complete~\cite{DGP09}.

%% file: three_prisoners.tex
\begin{table}[t]
\centering
\begin{tabular}{cccc}
$a$                                 & $u_1$                  & $u_2$                  & $u_3$                  \\ \hline
\multicolumn{1}{|c|}{$(c_1, c_2, c_3)$} & \multicolumn{1}{c|}{7} & \multicolumn{1}{c|}{7} & \multicolumn{1}{c|}{7} \\ \hline
\multicolumn{1}{|c|}{$(c_1, c_2, d_3)$} & \multicolumn{1}{c|}{3} & \multicolumn{1}{c|}{3} & \multicolumn{1}{c|}{9} \\ \hline
\end{tabular}
\begin{tabular}{cccc}
$a$                                 & $u_1$                  & $u_2$                  & $u_3$                  \\ \hline
\multicolumn{1}{|c|}{$(c_1, d_2, c_3)$} & \multicolumn{1}{c|}{3} & \multicolumn{1}{c|}{9} & \multicolumn{1}{c|}{3} \\ \hline
\multicolumn{1}{|c|}{$(c_1, d_2, d_3)$} & \multicolumn{1}{c|}{0} & \multicolumn{1}{c|}{5} & \multicolumn{1}{c|}{5} \\ \hline
\end{tabular}
\begin{tabular}{cccc}
$a$                                 & $u_1$                  & $u_2$                  & $u_3$                  \\ \hline
\multicolumn{1}{|c|}{$(d_1, c_2, c_3)$} & \multicolumn{1}{c|}{9} & \multicolumn{1}{c|}{3} & \multicolumn{1}{c|}{3} \\ \hline
\multicolumn{1}{|c|}{$(d_1, c_2, d_3)$} & \multicolumn{1}{c|}{5} & \multicolumn{1}{c|}{0} & \multicolumn{1}{c|}{5} \\ \hline
\end{tabular}
\begin{tabular}{cccc}
$a$                                 & $u_1$                  & $u_2$                  & $u_3$                  \\ \hline
\multicolumn{1}{|c|}{$(d_1, d_2, c_3)$} & \multicolumn{1}{c|}{5} & \multicolumn{1}{c|}{5} & \multicolumn{1}{c|}{0} \\ \hline
\multicolumn{1}{|c|}{$(d_1, d_2, d_3)$} & \multicolumn{1}{c|}{1} & \multicolumn{1}{c|}{1} & \multicolumn{1}{c|}{1} \\ \hline
\end{tabular}
\vspace{0.2cm}
\caption{Utilities for an instance of a three-player prisoner's dilemma.}
\label{tab:3pd}
\vspace{-.75cm}
\end{table}

%% file: experiments.tex
\section{Case Studies and Experimental Results}\label{expts-sec}

We have implemented our approach on top of PRISM-games 3.0~\cite{KNPS20}, extending the implementation to support multi-coalitional equilibria-based properties. The files for the case studies and results in this section are available from~\cite{files}.

\startpara{Implementation} CSGs are specified using the PRISM-games 3.0 modelling language, as described in~\cite{KNPS19,KNPS20}. Models are built and stored using the tool's Java-based `explicit' engine, which employs sparse matrices. Finding SWNE of NFGs, which can be reduced to solving a nonlinear programming problem (see \sectref{mc-sect}), is performed using a combination of the SMT solver Z3~\cite{Z3} and the nonlinear optimisation suite {\sc Ipopt}~\cite{Wac09}. Although SMT solvers are able to find solutions to nonlinear problems, they are not guaranteed to do so and are only efficient in certain cases. These cases include when there is a small number of actions per player or finding support assignments for which an equilibrium is not possible. To mitigate the inefficiencies of the SMT solver, we use Z3 for filtering out unsatisfiable support assignments with a timeout: given a support assignment, Z3 returns either \emph{unsat}, \emph{sat} or \emph{unknown} (if the timeout is reached). If either \emph{sat} or \emph{unknown} are returned, then the assignment is passed to {\sc Ipopt}, which checks for satisfiability (if required) and computes SWNE values using an interior-point filter line-search algorithm~\cite{WB06}. To speed up the overall computation the support assignments are analysed in parallel. We also search for and filter out \emph{dominated strategies} as a precomputation step. The NFGs are built on the fly, as well as the gradient of the objective function \eqnref{opt-eqn} and the Jacobian of the constraints \eqnref{eq-eqn}--\eqnref{ge-eqn}, which are required as an input to {\sc Ipopt}.

\input{table_one_shot} 

Table~\ref{tab:smt-stats} presents experimental results for solving various NFGs (generated with GAMUT~\cite{NWSL04}) using Z3 (with a timeout of 20ms) and {\sc Ipopt}. For each NFG, the table lists the numbers of players, actions of each player and support assignments. The table also includes the supports of each type returned by Z3 and the solution time of {\sc Ipopt}. As can be seen, using Z3 significantly reduces the assignments {\sc Ipopt} needs to analyse, by orders of magnitude in some cases. However, as the number of actions grows, the number of assignments that remain for {\sc Ipopt} to solve increases rapidly, and therefore so does the solution time. Furthermore, increasing the number of players only magnifies this issue.

The results show that solving NFGs can be computationally very expensive.
Note that just finding an NE is already a difficult problem,
whereas we search for SWNE, and hence need to find \emph{all} NE.
For example, in \cite{PNS04}, using a backtracking search algorithm or either of the Simplicial Subdivision~\cite{LTH87} and the Govindan-Wilson~\cite{GW03} algorithms for finding a sample NE, there are instances of NFGs with 6 players and 5 actions that timeout after 30 minutes. %

We also comment that care needs to be taken with numerical computations.
The value iteration part of the model checking algorithm
is (as usual) implemented using floating point arithmetic,
and may therefore exhibit small rounding errors.
However, the intermediate results are passed to solvers,
which may expect inputs in terms of rational numbers (Z3 in this case).
It could be beneficial to investigate the use of arbitrary precision arithmetic instead. %

\vskip0.6em
We now present case studies and experimental results to demonstrate the applicability and performance of our approach and implementation.

\input{table}

\startpara{Efficiency and Scalability} Table~\ref{tab:results} presents a selection of results demonstrating the performance of the implementation. The models in the table are discussed in more detail below. The results were carried out using a 2.10GHz Intel Xeon Gold with 16GB of JVM memory. The table includes statistics for the models: number of players, states, (maximum) actions for each player in a state, transitions and the times to both build and verify the models. All models have been verified in under 2 hours and in most cases much less than this. The largest model, verified in under 15 minutes, has 4 players, almost 1.5 million states and 5 million transitions. The majority of the time is spent solving NFG games and, as shown in Table~\ref{tab:smt-stats}, this varies depending on the number of choices and players.

\startpara{Secret Sharing} The first case study is the \emph{secret sharing} protocol of~\cite{HT04}, which uses uncertainty to induce cooperation. The protocol is defined for 3 agents and can be extended to more agents by partitioning the agents into three groups. Since the 3 agents act independently, this protocol could not be analysed with the two-coalitional variant of rPATL~\cite{KNPS19}. Each agent has an unfair coin with the same bias ($\alpha$). In the first step of the protocol, agents flip their coins, and if their coins land on heads, they are supposed to send their share of the secret to the other agents. In the second step, everyone reveals the value of their coin to the other agents. The game ends if all agents obtain all shares and therefore can all reconstruct the secret, or an agent cheats, i.e., fails to send their share to another agent when they are supposed to. %
If neither of these conditions hold, new shares are issued to the agents and a new round starts. The protocol assumes that each agent prefers to learn the secret and that others do not learn. This is expressed by the utilities $u_{3}$, $u_{2}$, $u_{1}$ and $u_{0}$ that an agent $i$ gets if all the agents, two agents (including $i$), only $i$ and no agent is able to learn the secret, respectively. %

A \emph{rational} agent in this context is one that has the choice of cheating and ignoring the coin toss in order to maximise their utility. An \emph{altruistic} agent is one who strictly follows the protocol and a \emph{byzantine} agent has a probability ($p_\mathit{fail}$) of failing and subsequently sending or computing the wrong values. \figref{fig::ssraa_rba} presents the expected utilities when there are two altruistic and one rational agent and when there is one altruistic, one byzantine and one rational agent as $\alpha$ varies. The results when there is one altruistic and two rational agents or three rational agents yield the same graph as
\figref{fig::ssraa_rba}(a), where the one or two additional rational agents utilities match those of the altruistic agents. According to the theoretical results of \cite{HT04}, for a model with one rational and two altruistic agents, the rational agent only has an incentive to cheat if:
\begin{equation}\label{ss-eqn}
(u_{1}{\cdot}\alpha^2 + u_{0}{\cdot}(1-\alpha)^2)/(\alpha^2+(1-\alpha)^2) > u_{3} \, .
\end{equation}
This result is validated by \figref{fig::ssraa_rba}(a) for the given utility values; the rational agent only cheats when $\alpha {\geq} 0.5$ (for $\alpha{<}0.5$ all agents receive a utility of 1 corresponding to all agents getting the secret), which corresponds to when \eqnref{ss-eqn} holds for our chosen utility values. Furthermore, \figref{fig::ssraa_rba} also shows that the closer $\alpha$ is to one then the greater the expected utility of a rational agent. \figref{fig::ssraa_rba}(b) also shows that, with a byzantine agent, the rational agent cheats when $\alpha{\geq}0.4$. 

\figref{fig::ssrra_rrr_rounds} plots the expected utilities of the agents when the protocol stops after a maximum number of rounds ($r_\mathit{max}$) when $\alpha{=}0.3$ and $\alpha{=}0.8$. The utilities converge more slowly for $\alpha{=}0.3$, since, when $\alpha$ is small,  there is a higher chance that an agent flips tails in a round, meaning not all agents will share their secret in this round and the protocol will move into another round. Again we see that there are more incentives for a rational agent to cheat as $\alpha$ gets closer to 1. However, when $\alpha{=}0.3$ and there are altruistic agents, the incentive decreases and eventually disappears as the number of rounds increases.

\begin{figure}[t]
\centering
\input{scrt_shrng_raa.tex}
\input{scrt_shrng_rba_pf02.tex}
\vspace*{-0.2cm}
\caption{$\nashop{\mathit{usr}_1{:}\mathit{usr}_2{:}\mathit{usr}_3}{\max=?}{\rewop{}{}{\future \mathsf{done}}{+}\rewop{}{}{\future \mathsf{done}}{+}\rewop{}{}{\future \mathsf{done}}}$ ($p_{fail}{=}0.2$).}\label{fig::ssraa_rba}
\vspace*{-0.2cm}
\end{figure}

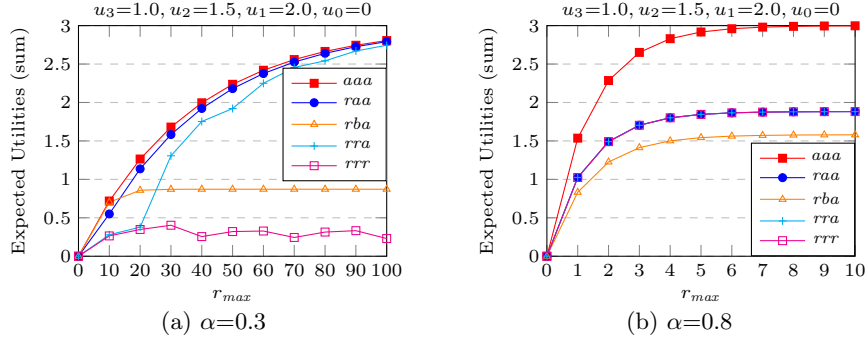
\begin{figure}[t]
\centering
\input{scrt_shrng_a03_sum_v2.tex}
\input{scrt_shrng_a08_sum.tex}
\vspace*{-0.2cm}
\caption{Expected utilities over a bounded number of rounds ($p_{fail}{=}0.2$ for $\mathit{rba}$).}\label{fig::ssrra_rrr_rounds}
\vspace*{-0.6cm}
\end{figure}

\startpara{Public Good Game} We consider a variant of the public good game presented in \egref{csg-eg}, in which the parameter $f$ is fixed, where each player receives an initial amount of capital ($e_\mathit{init}$) and, in each of $k$ months, can invest none, half or all of their current capital. A 2-player version of the game was modelled in~\cite{KNPS20}.

\begin{figure}[t]
\centering
\input{pblc_gd_3act_ind.tex}
\input{pblc_gd_3act_sum.tex}
\vspace*{-0.2cm}
\caption{$\nashop{p_1{:}p_2{:}p_3}{\max=?}{\rewop{c_1}{}{\sinstant{=r_\mathit{max}}}{+}\rewop{c_2}{}{\sinstant{=r_\mathit{max}}}{+}\rewop{c_3}{}{\sinstant{=r_\mathit{max}}}}$~($e_\mathit{init}{=}5$, $k{=}3$).}\label{fig:pg}
\vspace*{-0.6cm}
\end{figure}

\figref{fig:pg} presents results for the 3-player public good game as $f$ varies, plotting the expected utilities when the players act in isolation and, for comparison, when player 1 acts in isolation and players 2 and 3 form a coalition (indicated by $\coalition{}$), which would be required if the two-coalitional variant of rPATL~\cite{KNPS19} was used. When the players act in isolation, if $f{\leq}2$, then there is no incentive for the players to invest. As $f$ increases, the players start to invest some of their capital in some of the months, and when $f{=}3$ each player invests all their capital in each month. On the other hand, when players 2 and 3 act in a coalition, there is incentive to invest capital for smaller values of $f$, as players 2 and 3 can coordinate their investments to ensure they both profit; however, player 1 also gains from these investments, and therefore has no incentive to invest in the final month. As $f$ increases, there is a greater incentive for player 1 to invest and the final capital for all the players increases. The drop in the capital of player 1, as $f$ increases, is caused by players 2 and 3 coordinating against player 1 and decreasing their investments. This forces player 1 to invest to increase its investment which, as profits are shared, also increases the capital of players 2 and 3.

\startpara{Aloha} This case study concerns a number of users trying to send packets using the slotted ALOHA protocol introduced in~\cite{KNPS19}. In a time slot, if a single user tries to send a packet, there is a probability $q$ that the packet is sent; if $k$ users try and send, then the probability decreases to $q/k$. If sending a packet fails, the number of slots a user waits before resending is set according to an exponential backoff scheme. The analysis of the model in~\cite{KNPS19} consisted of considering three users with two acting in coalition. We extend the analysis by considering the case when the three act in isolation and extend the model with a fourth user. The objectives concern maximising the probability of sending a packet within a deadline, e.g.\ $\nashop{\mathit{usr}_1{:}\cdots{:}\mathit{usr}_m}{\max=?}{\probop{}{\future (\mathsf{s}_1 {\wedge} t{\leq}D)}{+}{\cdots}{+}\probop{}{\future (\mathsf{s}_m {\wedge} t{\leq}D)}}$, and the expected time to send a packet. By allowing the users to act independently we find that the expected time required for all users to send their packets reduces compared to when two of the players act as a coalition.

\startpara{Medium Access Control} This case study is based on a deterministic concurrent game model of medium access control~\cite{BRE13}. The model consists of two users that have limited energy and share a wireless channel. The users repeatedly choose to transmit or wait and, if both transmit, the transmissions fail due to interference. We previously extended the model to three users and added the probability of transmissions failing (which is dependent on the number of users transmitting)~\cite{KNPS19}. However, the analysis was restricted to the scenario where two users were in coalition~\cite{KNPS19}. We can now remove this restriction and analyse the case when each user tries to maximise the expected number of messages they send over a bounded number of steps and extend this analysis to four users.  %

%% file: table_one_shot.tex
\begin{table}[t]
\centering
{\scriptsize
\begin{tabular}{|c|r|r|r||r|r|r|r|} \hline
\multirow{2}{*}{Game} & 
\multicolumn{1}{c|}{\multirow{2}{*}{Players}}
& \multicolumn{1}{c|}{\multirow{2}{*}{Actions}} & \multicolumn{1}{c||}{\multirow{2}{*}{Supports}} &
 \multicolumn{3}{c|}{Supports returned by Z3} &
 \multicolumn{1}{c|}{Time(s)} \\ \cline{5-7}
& & &  & \multicolumn{1}{c|}{\emph{$\;\;$unsat$\;\;$}} & \multicolumn{1}{c|}{\emph{$\;\;$sat$\;\;$}} & \multicolumn{1}{c|}{\emph{unknown}} & \multicolumn{1}{c|}{{\sc Ipopt}}
 \\ \hline \hline	
\multirow{8}{*}{\emph{Majority Voting}}
&  3 & 3,3,3     &    343 &    330 &  12 &     1 & 0.309    \\
&  3 & 4,4,4     &  3,375 &  3,236 & 110 &    29 & 18.89    \\
&  3 & 5,5,5     & 29,791 & 26,250 & 155 & 3,386 & 336.5  \\
&  4 & 2,2,2,2   &     81 &     59 &  22 &     0 & 0.184    \\
&  4 & 3,3,3,3   &  2,401 &  2,212 &  87 &   102 & 6.847    \\
&  4 & 4,4,4,4   & 50,625 & 41,146 & 518 & 8,961 & 1,158 \\
&  5 & 2,2,2,2,2 &    243 &    181 &  62 &     0 & 0.591    \\
&  5 & 3,3,3,3,3 & 16,807 & 14,950 & 266 & 1,591 & 253.3  \\
\hline\hline
\multirow{7}{*}{\emph{Covariant game}}
&  3 & 3,3,3     &    343 &    304 &  6 &     33 & 7.645    \\
&  3 & 4,4,4     &  3,375 &  2,488 & 16 &    871 & 203.8  \\
&  3 & 5,5,5     & 29,791 & 14,271 &  8 & 15,512 & 5,801 \\
&  4 & 2,2,2,2   &     81 &     76 &  3 &      2 & 0.106    \\
&  4 & 3,3,3,3   &  2,401 &  1,831 &  0 &    570 & 183.0  \\
&  5 & 2,2,2,2,2 &    243 &    221 &  8 &     14 & 4.128    \\
&  5 & 3,3,3,3,3 & 16,807 &  6,600 &  7 & 10,200 & 5,002 \\
\hline
\end{tabular}}
\vspace*{0.1cm}
\caption{Finding SWNE in NFGs (timeout of 20ms for Z3).}
\label{tab:smt-stats}
\vspace*{-0.6cm}
\end{table}

%% file: table.tex
\begin{table}[!t]
\centering
{\scriptsize
\begin{tabular}{|c|c|c||r|r|r||r|r|} \hline
\multicolumn{1}{|c|}{Case study \& property} & 
\multicolumn{1}{c|}{\multirow{2}{*}{Players}} & 
\multicolumn{1}{c||}{Param.} & 
\multicolumn{3}{c||}{CSG statistics} & \multicolumn{1}{c|}{Constr.} & \multicolumn{1}{c|}{Verif.} \\ 
\cline{4-6}
\multicolumn{1}{|c|}{[parameters]} &
&
\multicolumn{1}{c||}{values} & 
\multicolumn{1}{c|}{States} & 
\multicolumn{1}{c|}{Max. Act.} & 
\multicolumn{1}{c||}{Trans.} & \multicolumn{1}{c|}{time(s)} &  \multicolumn{1}{c|}{time (s)} \\ \hline \hline
\multirow{4}{*}{\shortstack[c]{\emph{Secret Sharing} \\
$\rewop{}{\max=?}{\future \mathsf{d}{\vee}\mathit{r}{=}r_\mathit{max}}$ \\
$\mathit{model}$/[$\alpha$,$r_\mathit{max}$,$p_{fail}$]}} &
\multirow{4}{*}{3} 
&  \emph{raa}/0.3,10,\hrulefill  & 4,279 &   2,1,1 &  5,676 & 0.057 & 0.565 \\
&& \emph{rba}/0.3,10,0.2        &  7,095 &   2,1,1 &  9,900 & 0.090 & 0.939 \\
&& \emph{rra}/0.3,10,\hrulefill &  8,525 &   2,2,1 & 11,330 & 0.250 & 25.79 \\
&& \emph{rrr}/0.3,10,\hrulefill & 17,017 &   2,2,2 & 22,638 & 0.250 & 96.07 \\
\hline\hline
\multirow{8}{*}{\shortstack[c]{\emph{Public Good} \\
$\rewop{}{\max=?}{\sinstant{=k_{\scale{.75}{max}}}}$ \\
$\mbox{[$f,k_{max}$]}$}} &
\multirow{3}{*}{3}
&  2.9,2 &     758 & 3,3,3 &   1,486 & 0.098 &  7.782 \\
&& 2.9,3 &  16,337 & 3,3,3 &  36,019 & 0.799 &  110.1 \\
&& 2.9,4 & 279,182 & 3,3,3 & 703,918 & 6.295 &  1,459 \\
\cline{2-8}
&\multirow{3}{*}{4}
&  2.9,1 &      83 & 3,3,3,3 &     163 & 0.046 & 0.370 \\
&& 2.9,2 &   6,644 & 3,3,3,3 &  13,204 & 0.496 & 7.111 \\
&& 2.9,3 & 399,980 & 3,3,3,3 & 931,420 & 11.66 & 99.86 \\
\cline{2-8}
&\multirow{2}{*}{5}
&  2.9,1 &    245 & 3,3,3,3,3 &     487 & 0.081 & 2.427 \\
&& 2.9,2 & 59,294 & 3,3,3,3,3 & 118,342 & 2.572 & 2,291 \\
\hline\hline
\multirow{8}{*}{\shortstack[c]{\emph{Aloha (deadline)} \\
$\probop{\max=?}{\future \mathsf{s}_i {\wedge} t{\leq}D}$ \\
$\mbox{[$b_{max},D$]}$}} &
\multirow{4}{*}{3}
&  1,8 &   3,519 & 2,2,2 &     5,839 & 0.168 &  11.23 \\
&& 2,8 &  14,230 & 2,2,2 &    28,895 & 0.430 &  14.05 \\
&& 3,8 &  72,566 & 2,2,2 &   181,438 & 1.466 &  18.41 \\
&& 4,8 & 413,035 & 2,2,2 & 1,389,128 & 7.505 &  43.23 \\
\cline{2-8}
&
\multirow{3}{*}{4}
&  1,8 &    23,251 & 2,2,2,2 &    42,931 & 0.708 & 75.59 \\
&& 2,8 &   159,892 & 2,2,2,2 &   388,133 & 3.439 & 131.7 \\
&& 3,8 & 1,472,612 & 2,2,2,2 & 4,777,924 & 28.69 & 819.2 \\
\cline{2-8}
& \multirow{1}{*}{5}
&  1,8 & 176,777 & 2,2,2,2,2 & 355,209 & 3.683 & 466.3 \\
\hline \hline
\multirow{4}{*}{\shortstack[c]{\emph{Aloha} \\
$\rewop{}{\min=?}{\future \mathsf{s}_i}$ \\
$\mbox{[$b_{max}$]}$}} &
\multirow{4}{*}{3}
&  1 &    1,034 & 2,2,2 &   1,777 & 0.096 & 40.76 \\
&& 2 &    5,111 & 2,2,2 &  10,100 & 0.210 & 29.36 \\
&& 3 &   22,812 & 2,2,2 &  56,693 & 0.635 & 51.22 \\
&& 4 &  107,799 & 2,2,2 & 355,734 & 2.197 & 150.1 \\
\hline\hline
\multirow{4}{*}{\shortstack[c]{\emph{Medium access}  \\
$\rewop{}{\max=?}{\scumul{\leq k}}$ \\
$\mbox{[$e_{max},k$]}$}} 
&
\multirow{3}{*}{3}
&   $\;$5,10 &  1,546 & 2,2,2 &   17,100 & 0.324 &  147.9 \\
&&     10,10 & 10,591 & 2,2,2 &  135,915 & 1.688 &  682.7 \\
&&     15,20 & 33,886 & 2,2,2 &  457,680 & 4.663 &  6,448 \\
\cline{2-8}
& \multirow{1}{*}{4}
&  $\;$5,5 & 15,936 & 2,2,2,2 & 333,314 & 4.932 & 3,581 \\
\hline
\end{tabular}}
\vspace*{0.1cm}
\caption{Statistics for a representative set of CSG verification instances.}\label{tab:results}
\vspace*{-1cm}
\end{table}

%% file: scrt_shrng_raa.tex
\begin{subfigure}{.49\textwidth}
\centering
\scriptsize{
\begin{tikzpicture}
\begin{axis}[
    title style={yshift=-2ex},
    title={$u_3{=}1.0,u_2{=}1.5,u_1{=}2.0,u_0{=}0.0$},
    ylabel={Expected Utilities},
    xlabel={$\alpha$},
    xmin=0.1, xmax=1.0,
    xtick={0.1,0.2,0.3,0.4,0.5,0.6,0.7,0.8,0.9,1.0},
    ymin=0.0, ymax=2.0,
    ytick={0.0,0.2,0.4,0.6,0.8,1.0,1.2,1.4,1.6,1.8,2.0},
    ymajorgrids=true,
    grid style=dashed,
    height=4.65cm,
    width=0.95\textwidth,
    legend entries={
                {$\mathit{rational}$},
                {$\mathit{altruistic}$},
                {$\mathit{altruistic}$}
                },
    legend style={at={(0,1)},
                anchor=north west, 
                nodes={scale=0.9, transform shape}}            
]
\addlegendimage{mark=square*,red,mark size=1.5pt}
\addlegendimage{mark=*,blue,mark size=1.5pt}
\addlegendimage{mark=triangle,orange,mark size=1.5pt}
]

\addplot[mark=square*,red,mark size=1.5pt] table [x=alpha, y=p1, col sep=comma]{scrt_shrng_raa.txt};
\addplot[mark=*,blue,mark size=1.5pt] table [x=alpha, y=p2, col sep=comma]{scrt_shrng_raa.txt};
\addplot[mark=triangle,orange,mark size=1.5pt] table [x=alpha, y=p3, col sep=comma]{scrt_shrng_raa.txt};

\end{axis}
\end{tikzpicture}
}
\vspace*{-0.2cm}
\caption{2 altruistic and 1 rational}
\end{subfigure}

%% file: scrt_shrng_rba_pf02.tex
\begin{subfigure}{.49\textwidth}
\centering
\scriptsize{
\begin{tikzpicture}
\begin{axis}[
    title style={yshift=-2ex},
    title={$u_3{=}1.0,u_2{=}1.5,u_1{=}2.0,u_0{=}0.0$},
    ylabel={Expected Utilities},
    xlabel={$\alpha$},
    xmin=0.1, xmax=1.0,
    xtick={0.1,0.2,0.3,0.4,0.5,0.6,0.7,0.8,0.9,1.0},
    ymin=0.0, ymax=2.0,
    ytick={0.0,0.2,0.4,0.6,0.8,1.0,1.2,1.4,1.6,1.8,2.0},
    ymajorgrids=true,
    grid style=dashed,
    height=4.65cm,
    width=0.95\textwidth,
    legend entries={
                {$\mathit{rational}$},
                {$\mathit{byzantine}$},
                {$\mathit{altruistic}$}
                },
    legend style={at={(0.0,1.0)},
                anchor=north west, 
                nodes={scale=0.9, transform shape}}            
]
\addlegendimage{mark=square*,red,mark size=1.5pt}
\addlegendimage{mark=*,blue,mark size=1.5pt}
\addlegendimage{mark=triangle,orange,mark size=1.5pt}
]

\addplot[mark=square*,red,mark size=1.5pt] table [x=alpha, y=p1, col sep=comma]{scrt_shrng_rba_pf02.txt};
\addplot[mark=*,blue,mark size=1.5pt] table [x=alpha, y=p2, col sep=comma]{scrt_shrng_rba_pf02.txt};
\addplot[mark=triangle,orange,mark size=1.5pt] table [x=alpha, y=p3, col sep=comma]{scrt_shrng_rba_pf02.txt};

\end{axis}
\end{tikzpicture}
}
\vspace*{-0.2cm}
\caption{1 altruistic, 1 byzantine and 1 rationalå}
\end{subfigure}

%% file: scrt_shrng_a03_sum_v2.tex
\begin{subfigure}{0.49\textwidth}
\centering
\scriptsize{
\begin{tikzpicture}
\begin{axis}[
    title style={yshift=-2ex},
    title={$u_3{=}1.0,u_2{=}1.5,u_1{=}2.0,u_0{=}0$},
    ylabel={Expected Utilities (sum)},
    xlabel={$r_\mathit{max}$},
    xmin=0, xmax=100,
    xtick={0,10,20,30,40,50,60,70,80,90,100},
    ymin=0.0, ymax=3,
    ytick={0,0.5,1.0,1.5,2.0,2.5,3.0},
     ymajorgrids=true,
    grid style=dashed,
    height=4.65cm,
    width=0.95\textwidth,
    legend entries={
                {$\mathit{aaa}$},
                {$\mathit{raa}$},
                {$\mathit{rba}$},
                {$\mathit{rra}$},
                {$\mathit{rrr}$},
                },
    legend style={at={(1.0,0.325)},
                anchor=south east, 
                nodes={scale=0.9, transform shape}}            
]
\addlegendimage{mark=square*,red,mark size=1.5pt}
\addlegendimage{mark=*,blue,mark size=1.5pt}
\addlegendimage{mark=triangle,orange,mark size=1.5pt}
\addlegendimage{mark=+,cyan,mark size=1.5pt}
\addlegendimage{mark=square,magenta,mark size=1.5pt}
]

\addplot[mark=square*,red,mark size=1.5pt] table [x=rmax, y=aaa, col sep=comma]{scrt_shrng_a03_sums_v2.txt};
\addplot[mark=*,blue,mark size=1.5pt] table [x=rmax, y=raa, col sep=comma]{scrt_shrng_a03_sums_v2.txt};
\addplot[mark=triangle,orange,mark size=1.5pt] table [x=rmax, y=rba, col sep=comma]{scrt_shrng_a03_sums_v2.txt};
\addplot[mark=+,cyan,mark size=1.5pt] table [x=rmax, y=rra, col sep=comma]{scrt_shrng_a03_sums_v2.txt};
\addplot[mark=square,magenta,mark size=1.5pt] table [x=rmax, y=rrr, col sep=comma]{scrt_shrng_a03_sums_v2.txt};

\end{axis}
\end{tikzpicture}
}
\vspace*{-0.2cm}
\caption{$\alpha{=}0.3$}

\end{subfigure}

%% file: scrt_shrng_a08_sum.tex
\begin{subfigure}{0.49\textwidth}
\centering
\scriptsize{
\begin{tikzpicture}
\begin{axis}[
    title style={yshift=-2ex},
    title={$u_3{=}1.0,u_2{=}1.5,u_1{=}2.0,u_0{=}0$},
    ylabel={Expected Utilities (sum)},
    xlabel={$r_\mathit{max}$},
    xmin=0, xmax=10,
    xtick={0,1,2,3,4,5,6,7,8,9,10},
    ymin=0.0, ymax=3,
    ytick={0,0.5,1.0,1.5,2.0,2.5,3.0},
    ymajorgrids=true,
    grid style=dashed,
    height=4.65cm,
    width=0.95\textwidth,
    legend entries={
                {$\mathit{aaa}$},
                {$\mathit{raa}$},
                {$\mathit{rba}$},
                {$\mathit{rra}$},
                {$\mathit{rrr}$},
                },
    legend style={at={(1.0,0.0)},
                anchor=south east, 
                nodes={scale=0.9, transform shape}}            
]
\addlegendimage{mark=square*,red,mark size=1.5pt}
\addlegendimage{mark=*,blue,mark size=1.5pt}
\addlegendimage{mark=triangle,orange,mark size=1.5pt}
\addlegendimage{mark=+,cyan,mark size=1.5pt}
\addlegendimage{mark=square,magenta,mark size=1.5pt}
]

\addplot[mark=square*,red,mark size=1.5pt] table [x=rmax, y=aaa, col sep=comma]{scrt_shrng_a08_sums.txt};
\addplot[mark=*,blue,mark size=1.5pt] table [x=rmax, y=raa, col sep=comma]{scrt_shrng_a08_sums.txt};
\addplot[mark=triangle,orange,mark size=1.5pt] table [x=rmax, y=rba, col sep=comma]{scrt_shrng_a08_sums.txt};
\addplot[mark=+,cyan,mark size=1.5pt] table [x=rmax, y=rra, col sep=comma]{scrt_shrng_a08_sums.txt};
\addplot[mark=square,magenta,mark size=1.5pt] table [x=rmax, y=rrr, col sep=comma]{scrt_shrng_a08_sums.txt};

\end{axis}
\end{tikzpicture}
}
\vspace*{-0.2cm}
\caption{$\alpha{=}0.8$}

\end{subfigure}

%% file: pblc_gd_3act_ind.tex
\begin{subfigure}{.49\textwidth}
\centering
\scriptsize{
\begin{tikzpicture}
\begin{axis}[
    title style={yshift=-2ex},
    ylabel={Expected Capital},
    xlabel={$f$},
    xmin=1.5, xmax=3.0,
    xtick={1.5,2.0,2.5,3.0},
    ymin=0, ymax=275,
    ytick={0,25,50,75,100,125,150,175,200,225,250,275},
    ymajorgrids=true,
    grid style=dashed,
    height=4.65cm,
    width=0.95\textwidth,
    legend entries={
                {$\coalition{p_1}$},
                {$\coalition{p_2,p_3}$},
                {$p_1$},
                {$p_2$},
                {$p_3$},
                },
    legend style={at={(0.0,1.0)},
                anchor=north west, 
                nodes={scale=0.9, transform shape}}            
]
\addlegendimage{mark=square*,red,mark size=1.5pt}
\addlegendimage{mark=*,blue,mark size=1.5pt}
\addlegendimage{mark=triangle,orange,mark size=1.5pt}
\addlegendimage{mark=+,cyan,mark size=1.5pt}
\addlegendimage{mark=square,magenta,mark size=1.5pt}
]

\addplot[mark=square*,red,mark size=1.5pt] table [x=f, y=2p1, col sep=comma]{pblc_gd_3act.txt};
\addplot[mark=*,blue,mark size=1.5pt] table [x=f, y=2p23, col sep=comma]{pblc_gd_3act.txt};
\addplot[mark=triangle,orange,mark size=1.5pt] table [x=f, y=3p1, col sep=comma]{pblc_gd_3act.txt};
\addplot[mark=+,cyan,mark size=1.5pt] table [x=f, y=3p2, col sep=comma]{pblc_gd_3act.txt};
\addplot[mark=square,magenta,mark size=1.5pt] table [x=f, y=3p3, col sep=comma]{pblc_gd_3act.txt};

\end{axis}
\end{tikzpicture}
}
\vspace*{-0.2cm}
\caption{Individual rewards.}
\end{subfigure}

%% file: pblc_gd_3act_sum.tex
\begin{subfigure}{.49\textwidth}
\centering
\scriptsize{
\begin{tikzpicture}
\begin{axis}[
    title style={yshift=-2ex},
    ylabel={Expected Capital (sum)},
    xlabel={$f$},
    xmin=1.5, xmax=3.0,
    xtick={1.5,2.0,2.5,3.0},
    ymin=0, ymax=405,
    ytick={0,50,100,150,200,250,300,350,400},
    ymajorgrids=true,
    grid style=dashed,
    height=4.65cm,
    width=0.95\textwidth,
    legend entries={
                {$\coalition{p_1}+\coalition{p_2,p_3}$},
                {$p_1+p_2+p_3$}
                },
    legend style={at={(0,1)},
                anchor=north west, 
                nodes={scale=0.9, transform shape}}            
]
\addlegendimage{mark=square*,green,mark size=1.5pt}
\addlegendimage{mark=*,purple,mark size=1.5pt}
]

\addplot[mark=square*,green,mark size=1.5pt] table [x=f, y=2psum, col sep=comma]{pblc_gd_3act.txt};
\addplot[mark=*,purple,mark size=1.5pt] table [x=f, y=3psum, col sep=comma]{pblc_gd_3act.txt};

\end{axis}
\end{tikzpicture}
}
\vspace*{-0.2cm}
\caption{Sum of rewards.}
\end{subfigure}

%% file: conclusions.tex
\section{Conclusions}

We have presented a logic and algorithm for model checking \emph{multi-coalitional} equilibria-based properties of CSGs, focusing on a variant of stopping games. We have implemented the approach in PRISM-games and demonstrated its applicability on a range of case studies and properties.
The main limitation of the approach is the time required for solving NFGs during value iteration as the number of players increases. Efficiency improvements that could be employed include filtering out conditionally dominated strategies~\cite{SW98}.
Future work will also include investigating correlated equilibria~\cite{Aum74} and mechanism design~\cite{NNGP09}.

\vskip0.2em
\startpara{Acknowledgements} 
This project has received funding from the European Research Council (ERC)
under the European Union’s Horizon 2020 research and innovation programme
(grant agreement No.~834115) 
and the EPSRC Programme Grant on Mobile Autonomy (EP/M019918/1).

%% file: correctness.tex
\section{Correctness of the Model Checking Algorithm for Nonzero-Sum State Formulae}\label{correctness-app}

We fix a game $\game$ and nonzero-sum state formula $\nashop{C_1{:}{\cdots}{:}C_m}{\opt \sim x}{\theta}$ and let $\cC = \{ C_1, \dots, C_m \}$. For the case of finite-horizon nonzero-sum formulae the correctness of the model checking algorithm follows from the fact that we use backward induction~\cite{SW+01,NMK+44}. Below we consider probabilistic and expected reachability objectives in the case that $\opt{=}\max$. The remaining cases for infinite-horizon nonzero-sum formulae follow similarly. For a nonzero-sum formula $\theta$, we denote by $\theta_i$ the $i$th term in $\theta$.
We first introduce the following objectives for the coalition game $\game^\cC$ which are $n$-step approximations of $X^\theta_{i}$.
\begin{definition}\label{bounded-objective-def}
For any probabilistic or expected reachability nonzero-sum formula $\theta$, $1{\leq}i{\leq}m$ and $n \in \Nset$, let $X^\theta_{i,n}$ be the objective where for any path $\pi$ of $\game^\cC:$
\begin{eqnarray*}
X^{\probop{}{\future \! \phi^1}{+}{\cdots}{+}\probop{}{\future \! \phi^m}}_{i,n}(\pi) &\ = \ & \left\{ \begin{array}{cl}
1 & \ \mbox{if $\exists k {\leq} n. \, \pi(k) \sat \phi^i$} \\
0 & \ \mbox{otherwise}
\end{array} \right. \\
X^{\rewop{r_{\scale{.75}{1}}}{}{\future \! \phi^1}{+}{\cdots}{+}\rewop{r_{\scale{.75}{m}}}{}{\future \! \phi^m}}_{i,n}(\pi) & \ = \ & \left\{ \begin{array}{cl}
\infty
& \ \mbox{if} \; \forall k \in \Nset . \, \pi(k) \notsat \phi^i \\
\mbox{$\sum_{k=0}^{k^i}$} r(\pi,k) & \ \mbox{if $k^i \leq n{-}1$} \\
0 & \ \mbox{otherwise}
\end{array} \right.
\end{eqnarray*}
$r(\pi,k)=r_A(\pi(k),\pi[k])+r_S(\pi(k))$ and $k^i = \min \{ k{-}1 \mid k \in \Nset \wedge \pi(k) \sat \phi^i \}$.
\end{definition}
The following lemma demonstrates that, for a fixed strategy profile and state, the values of these objectives are non-decreasing and converge uniformly to the values of $\theta$.
\begin{lemma}\label{epsilon-lem}
For any probabilistic or expected reachability nonzero-sum formula $\theta$ we have that the sequence $\langle \Eset^{\sigma}_{\game^\cC,s}(X^\theta_i) \rangle_{n \in \Nset}$ is non-decreasing and, for any $\varepsilon{>}0$, there exists $N \in \Nset$ such that for any $n {\geq} N$, $s \in S$, $\sigma \in \Sigma_{\game^\cC}$ and $1{\leq}i{\leq}m:$
\[
0 \ \leq \ \Eset^{\sigma}_{\game^\cC,s}(X^\theta_i) - \Eset^{\sigma}_{\game^\cC,s}(X^\theta_{i,n}) \ \leq \ \varepsilon \, .
\]
\end{lemma}
\begin{proof}
Consider any probabilistic or expected reachability nonzero-sum formula $\theta$, state $s$ and $1{\leq} i {\leq} m$. Using \assumref{game-assum} we have that for subformulae $\rewop{r}{}{\future \phi^i}$, the set $\Sat(\phi^i)$ is reached with probability 1 from all states of $\game$ under all profiles, and therefore $\Eset^{\sigma}_{\game^\cC,s}(X^\theta_i)$ is finite. Furthermore, for any $n \in N$, by \defdefref{sem-def}{bounded-objective-def} we have that $\Eset^{\sigma}_{\game^\cC,s}(X^\theta_{i,n})$ is the value of state $s$ for the $n$th iteration of value iteration~\cite{CH08} when computing $\Eset^{\sigma}_{\game^\cC,s}(X^\theta_i)$ in the DTMC obtained from $\game^\cC$ by following the strategy $\sigma$. It therefore follows that the sequence is both non-decreasing and converges uniformly. \qed
\end{proof}
In the proof of correctness we will use the fact that $n$ iterations of value iteration is equivalent to performing backward induction on the following game trees.
\begin{definition}\label{trees-def}
For any state $s$ and $n \in \Nset$, let $\game^\cC_{n,s}$ be the game tree corresponding to playing $\game^\cC$ for $n$ steps when starting from state $s$ and then terminating. 
\end{definition}
We can map any strategy profile $\sigma$ of $\game^\cC$ to a strategy profile of $\game^\cC_{n,s}$ by only considering the choices of the profile over the first $n$ steps when starting from state $s$. This mapping is clearly surjective, i.e.\ we can generate all profiles of $\game^\cC_{n,s}$, but is not injective.
We require the following lemma relating the values of the objectives $X^\theta_{i,n}$ and $X^\theta_i$ over $\game^\cC$ and $\game^\cC_{n,s}$ for any $n \in \Nset$ and $s \in S$.
\begin{lemma}\label{precomp-lem}
For any probabilistic or expected reachability nonzero-sum formula $\theta$, state $s$ of $\game^\cC$, strategy profile $\sigma$, $n \in \Nset$ and $1{\leq}i{\leq}m:$ $\Eset^{\sigma}_{\game^\cC,s}(X^\theta_{i,n}) = 
\Eset^{\sigma}_{\game^\cC_{n,s}}(X^\theta_i)$.
\end{lemma}
\begin{proof}
The proof follows from \defdefref{bounded-objective-def}{trees-def}, in particular, we have that $X^\theta_{i,n}$ is the $n$-step approximations of $X^\theta_{i}$ and $\game^\cC_{n,s}$ corresponds to playing game $\game^\cC$ from state $s$ for $n$ steps. \qed
\end{proof}
We now define the strategy profiles synthesised during value iteration.
\begin{definition}\label{strats-def}
For any $n \in \Nset$ and $s \in S$, let $\sigma^{n,s}$ be the strategy profile generated for the game tree $\game^\cC_{n,s}$ (when considering value iteration as backward induction) and $\sigma^{n,\star}$ be the synthesised strategy profile for $\game^\cC$ after $n$ iterations.
\end{definition}
Before giving the proof of correctness we require the following results.
\begin{lemma}\label{backwards-lem}
For any state $s$ of $\game^\cC$, probabilistic or expected reachability nonzero-sum formula $\theta$ and $n \in \Nset$ we have that $\sigma^{n,s}$ is a subgame perfect SWNE of the CSG $\game^\cC_{n,s}$ for the objectives $(X^{\theta_1},\dots,X^{\theta_m})$.
\end{lemma}
\begin{proof}
The result follows from the fact that, for any $n \in \Nset$ and $s \in S$, the value iteration procedure selects SWNE, $n$ steps of value iteration corresponds to performing backward induction for the objectives $(X^{\theta_1},\dots,X^{\theta_m})$ in the game $\game^\cC_{n,s}$ and backward induction returns a subgame perfect NE~\cite{SW+01,NMK+44}. \qed
\end{proof}
The following proposition demonstrates that value iteration converges and depends on \assumref{game-assum}. Without this assumption convergence cannot be guaranteed as demonstrated by the counterexamples in~\cite{KNPS18b}. Although value iteration converges, unlike value iteration for MDPs or zero-sum games the generated sequence of values is not necessarily non-decreasing. 
\begin{proposition}\label{convergence-prop}
For any probabilistic or expected reachability nonzero-sum formula $\theta$ and state $s$, the sequence $\langle \V_{\game^\cC}(s,\theta,n) \rangle_{n \in \Nset}$ converges.
\end{proposition}
\begin{proof}
For any state $s$ and $n \in \Nset$ we can consider $\game^\cC_{n,s}$ as an $m$-player infinite-action NFG $\nfgame_{n,s}$ where for $1 {\leq} i {\leq} m$:
\begin{itemize}
\item
the set of actions of player $i$ equals the set of strategies of player $i$ in $\game^\cC$;
\item for the action pair $(\sigma_1,\sigma_2)$, the utility function for player $i$ returns $\Eset^{\sigma}_{\game^\cC_{n,s}}(X^\theta_i)$.
\end{itemize}
The correctness of this construction relies on the mapping of strategy profiles from the game $\game^\cC$ to $\game^\cC_{n,s}$ being surjective.  Using \lemref{epsilon-lem}, we have that the sequence $\langle \nfgame_{n,s} \rangle_{n \in N}$ of NFGs converges uniformly, and therefore, since $\V_{\game^\cC}(s,\theta,n)$ are subgame perfect SWNE values of $\game^\cC_{n,s}$ (see \lemref{backwards-lem}), the sequence $\langle \V_{\game^\cC}(s,\theta,n) \rangle_{n \in \Nset}$ also converges. \qed
\end{proof}
A similar convergence result to \propref{convergence-prop} has been shown for discounted properties of two-player games in~\cite{FL83}.
\begin{theorem}
For a given probabilistic or expected reachability nonzero-sum formula $\theta$ and $\varepsilon{>}0$, there exists $N \in \Nset$ such that for any $n {\geq} N$ the strategy profile $\sigma^{n,\star}$ is a subgame perfect $\varepsilon$-SWNE for $\game^\cC$ and the objectives $(X^{\theta_1},\dots,X^{\theta_m})$.
\end{theorem}
\begin{proof}
Consider any $\varepsilon{>}0$. From \defref{strats-def} for any $s\in S$, $n {\geq} \Nset$ and $1{\leq}i{\leq}m$:
\begin{equation}\label{1-eqn}
 \Eset^{\sigma^{n,\star}}_{\game^\cC,s}(X^\theta_i) = \Eset^{\sigma^{n,s}}_{\game^\cC_{n,s}}(X^\theta_i) \, .
\end{equation}
For any $k \in \Nset$ and $s \in S$, using \lemref{backwards-lem} we have that $\sigma^{k,s}$ is a NE of $\game^\cC_{k,s}$, and therefore for any $k \in \Nset$, $s\in S$ and $1{\leq}i{\leq}m$:
\begin{equation}\label{nash-eqn}
\Eset^{\sigma^{k,s}}_{\game^{C}_{k,s}}(X^\theta_i) 
\ = \
\sup\nolimits_{\sigma_i \in \Sigma_i^{\game^{\scale{.75}{C}}_{\scale{.75}{k,s}}}} \Eset^{\sigma^{k,s}_{-i}[\sigma_i]}_{\game^\cC_{k,s}}(X^\theta_i) \, .
\end{equation}
From \lemref{epsilon-lem} there exists $N \in \Nset$ such that for any $n {\geq} N$, $s \in S$ and $1{\leq}i{\leq}m$:
\begin{equation}\label{2-eqn}
\sup\nolimits_{\sigma_i \in \Sigma_i^{\game^{\scale{.75}{C}}}} \Eset^{\sigma^{n,\star}_{-i}[\sigma_i]}_{\game^\cC,s}(X^\theta_i) - \sup\nolimits_{\sigma_i \in \Sigma_i^{\game^{\scale{.75}{C}}}} \Eset^{\sigma^{n,\star}_{-i}[\sigma_i]}_{\game^\cC,s}(X^\theta_{i,n}) 
\ \leq \ \varepsilon \, .
\end{equation}
Therefore, for any $n {\geq} N$, $s \in S$ and $1{\leq}i{\leq}m$, using \eqnref{1-eqn} we have:
\begin{align*}
\Eset^{\sigma^{n,\star}}_{\game^\cC,s}(X^\theta_i) \ \ & = \ \ \Eset^{\sigma^{n,s}}_{\game^\cC_{n,s}}(X^\theta_i)   \\
& = \ \ \sup\nolimits_{\sigma_i \in \Sigma_i^{\game^{\scale{.75}{C}}_{\scale{.75}{n,s}}}} \Eset^{\sigma^{n,s}_{-i}[\sigma_i]}_{\game^\cC_{n,s}}(X^\theta_i)  & \mbox{by \eqnref{nash-eqn}} \\
& = \ \sup\nolimits_{\sigma_i \in \Sigma_i^{\game^{\scale{.75}{C}}}} \Eset^{\sigma^{n,\star}_{-i}[\sigma_i]}_{\game^\cC,s}(X^\theta_{i,n})  & \mbox{by \lemref{precomp-lem}} \\
& \leq \ \ \sup\nolimits_{\sigma_i \in \Sigma_i^{\game^{\scale{.75}{C}}}} \Eset^{\sigma^{n,\star}_{-i}[\sigma_i]}_{\game^\cC,s}(X^\theta_i) - \varepsilon & \mbox{by \eqnref{2-eqn} since $n {\geq} N$}
\end{align*}
and hence, since $\varepsilon{>}0$, $s \in S$ and $1{\leq}i{\leq}m$ were arbitrary, $\sigma^{n,\star}$ is a subgame perfect $\varepsilon$-NE. It remains to show that the strategy profile is a subgame perfect social welfare optimal $\varepsilon$-NE, which follows from the fact that when solving the bimatrix games during value iteration social welfare optimal NE are returned. \qed
\end{proof}